\documentclass[sts]{imsart}

\RequirePackage{amsthm,amsmath,amsfonts,amssymb}
\RequirePackage[numbers,sort&compress]{natbib}
\RequirePackage[colorlinks,citecolor=blue,urlcolor=blue]{hyperref}
\RequirePackage{graphicx}
\usepackage{mathtools}

\startlocaldefs
\theoremstyle{plain}

\newtheorem{theorem}{Theorem}[section]
\newtheorem{proposition}[theorem]{Proposition}

\theoremstyle{remark}

\newtheorem*{remark}{Remark}

\usepackage{algorithm}
\usepackage{algpseudocode}
\newcommand{\R}{\mathbb{R}}
\renewcommand{\O}{\mathcal{O}}
\newcommand{\T}{^\top}
\newcommand{\N}[1]{\mathcal{N}\left(#1\right)}
\newcommand{\lrp}[1]{\left(#1\right)}
\newcommand{\lrb}[1]{\left[#1\right]}
\newcommand{\lrs}[1]{\left\{#1\right\}}
\newcommand{\diag}[1]{\operatorname{diag}\left(#1\right)}
\newcommand{\tr}[1]{\operatorname{tr}\left(#1\right)}
\DeclareMathOperator*{\argmin}{arg\,min}
\DeclareMathOperator*{\argmax}{arg\,max}
\newcommand{\norm}[1]{\left\lVert#1\right\rVert}
\newcommand{\BIC}{\operatorname{BIC}}
\newcommand{\Skew}{\operatorname{Skew}}
\newcommand{\Fl}{\operatorname{Fl}}
\endlocaldefs

\begin{document}

\begin{frontmatter}
\title{The curse of isotropy: from principal components to principal subspaces}
\runtitle{The curse of isotropy}

\begin{aug}
\author[A]{\fnms{Tom}~\snm{Szwagier}\ead[label=e1]{tom.szwagier@inria.fr}\orcid{0000-0002-2903-551X}} 
\and
\author[B]{\fnms{Xavier}~\snm{Pennec}\ead[label=e2]{xavier.pennec@inria.fr}\orcid{0000-0002-6617-7664}} 


\address[A]{Tom Szwagier is PhD Candidate, Université Côte d'Azur and Inria, Sophia-Antipolis, France\printead[presep={\ }]{e1}.}

\address[B]{Xavier Pennec is Senior Research Scientist, Université Côte d'Azur and Inria, Sophia-Antipolis, France\printead[presep={\ }]{e2}.}

\end{aug}

\begin{abstract}
Principal component analysis is a ubiquitous tool in exploratory data analysis. It is widely used by applied scientists for visualization and interpretability purposes. We raise an important issue (the curse of isotropy) about the interpretation of principal components with close eigenvalues. They may indeed suffer from an important rotational variability, which is a pitfall for interpretation. Through the lens of a probabilistic covariance model parameterized with flags of subspaces, we show that the curse of isotropy cannot be overlooked in practice. In this context, we propose to transition from ill-defined principal components to more-interpretable principal subspaces. The final methodology (principal subspace analysis) is extremely simple and shows promising results on a variety of datasets from different fields.
\end{abstract}


\end{frontmatter}

\section{Introduction}\label{sec:intro}

Principal component analysis (PCA) \cite{jolliffe_principal_2002} is a universal method in data analysis. It gives the main modes of variation in the data by diagonalizing the empirical covariance matrix. The eigenvectors associated with the largest eigenvalues are the \textit{principal components}, and the subspace they span is used for dimension reduction and visualization.
Additionally, principal components can be used for exploratory data analysis and interpretability purposes. It has been extensively used on structured anatomical data (with components related to morphological features), in atmospheric sciences (with components related to climate patterns), computer vision (with so-called \textit{eigenfaces~\cite{sirovich_low-dimensional_1987}}) and many other fields. 
We refer to the chapters 4 and 11 of \cite{jolliffe_principal_2002} for detailed examples of principal component interpretation.

Let us assume that a dataset has been sampled from a multivariate Gaussian distribution.
If all the population covariance eigenvalues are \textit{simple} (i.e. distinct), then we can associate to each eigenvalue a unique eigenvector (up to sign {and scale}). Now, if some eigenvalues are \textit{multiple}, then those are associated with multidimensional eigenspaces, i.e. an infinite number of eigenvectors. This implies that the principal components associated with those multiple eigenvalues exhibit a large \textit{intersample variability}. More specifically, for any dataset size $n$, each independent $n$-sample from the distribution can yield totally different principal components, with a full rotational uncertainty within the eigenspaces. Therefore, under this multiple-eigenvalue assumption, principal components are unstable---regardless of $n$---which is fatal to data interpretability. We call this issue the \textit{curse~of~isotropy}.

In real datasets, empirical covariance eigenvalues are never exactly equal (they are almost surely different from a measure-theoretical point of view, cf. Theorem~\ref{appthm:PSA}), but some may be relatively close. In this case, it might be wiser to assume that \textit{close} eigenvalues are actually \textit{equal}---especially for small $n$---in order to avoid overfitting some spurious patterns caused by \textit{sampling errors}~\cite{north_sampling_1982}. Under this assumption, the dataset suffers from the curse of isotropy and one must be careful about interpreting the associated principal components. Therefore, identifying the curse of isotropy in practice boils down to answering the following question: \textit{when should we assume that the dataset has been sampled from a multivariate Gaussian distribution with repeated covariance eigenvalues?}

In this paper, we answer the question with an \textit{explicit} guideline, derived from two key concepts: \textit{parsimonious Gaussian modeling} and \textit{flags of subspaces}.
More specifically, we introduce a latent variable generative model called \textit{principal subspace analysis} (PSA). This model assumes a Gaussian density with repeated eigenvalues, where the sequence of eigenvalue multiplicities is specified by the so-called \textit{type} of the model. We show that PSA generalizes the celebrated Probabilistic PCA (PPCA) of Tipping and Bishop~\cite{tipping_probabilistic_1999} and unifies it with Isotropic PPCA (IPPCA)~\cite{bouveyron_high-dimensional_2007,bouveyron_intrinsic_2011}---a parsimonious version of PPCA suited to high dimensions. PSA models have a rich geometry relying on flag manifolds and stratify the space of covariance matrices. This enables us to assess the drop of model complexity caused by equalizing some eigenvalues and to perform efficient model selection based on parsimony-inducing criteria such as the Bayesian information criterion (BIC)---other criteria are investigated in Section~\ref{appsec:MS} with similar conclusions.
We show that two adjacent sample eigenvalues should be assumed equal when their \textit{relative eigengap} is lower than a given threshold. This threshold depends on $n$ but is independent of the dimension $p$.

The results are striking: in almost all the datasets that we analyze, the curse of isotropy arises.
This questions the numerous scientific works relying on the interpretation of principal components.
While this could sound fatal to exploratory data analysis, we show that the curse of isotropy can actually be leveraged to improve data interpretability.
Indeed, in such a situation, we suggest to give up \textit{principal components} and transition to more-interpretable \textit{principal subspaces}. Taking advantage of our generative model and factor rotation methods, we propose several qualitative and quantitative methods to increase the interpretability of principal components. We test the resulting PSA methodology on synthetic and real datasets and get promising results. More precisely, while principal components with close eigenvalues may be fuzzy{---as arbitrary linear combinations of latent variables---the principal subspace they span may contain more interpretable features.
}

\section{The curse of isotropy}

Let us consider a dataset sampled independently from a two-dimensional \textit{isotropic} Gaussian distribution. This implies that the eigenvalues of the population covariance matrix are equal. The sample covariance matrix, however, is an approximation of the population covariance matrix, whose accuracy improves with the number of observed samples~\cite{tyler_asymptotic_1981}. Notably, the empirical eigenvalues are almost surely distinct (cf. Theorem~\ref{appthm:PSA}). 
Therefore, PCA outputs the unique eigenvectors (up to sign) associated with each eigenvalue.
If we repeat this experiment several times independently and plot the principal components, we get Fig~\ref{fig:isotropy}.
\begin{figure}[t]
\centering
\includegraphics[width=\linewidth]{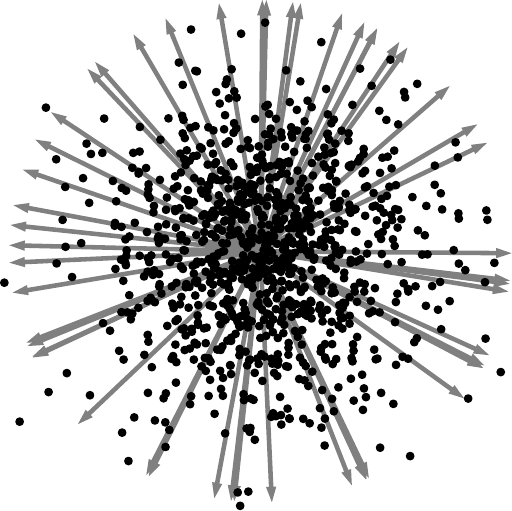}
\caption{Covariance eigenvectors of a dataset sampled from a two-dimensional isotropic Gaussian, repeated independently 25 times. Principal components have an isotropic intersample variability.}
\label{fig:isotropy}
\end{figure}
As we can see, the principal components are evenly spread in all directions---i.e. isotropically.  We call this phenomenon \textit{the curse of isotropy}. It is a curse since it yields principal components with high intersample variability and without any preferred direction. The observed components could therefore be random combinations of \textit{actually} interpretable components.

A legitimate question might then be: \textit{why (and when) should we assume that a given dataset has been sampled from a Gaussian distribution with repeated eigenvalues?} 
The Gaussian assumption is notably justified by the central limit theorem, the entropy maximization and the attractive computational properties that make Gaussian distributions the cornerstone of machine learning generative models~\cite{bishop_pattern_2006}.
Now, regarding the multiple-eigenvalue assumption, we have to go back to one of the founding principles of modeling that is the \textit{law of parsimony}, also known as \textit{Occam's razor}: ``The simplest explanation is usually the best one''. This principle is particularly applied in statistical modeling, where the limited number of observed samples makes overparameterized models overfitting~\cite{myung_counting_2000}. Notably, covariance matrices (which have ${O}(p^2)$ parameters) can almost never be correctly estimated in practice, especially in high dimensions~{\cite{pourahmadi_covariance_2011}}. Therefore, more parsimonious models have to be considered, like isotropic Gaussians (which have $1$ parameter---the variance), where all the covariance eigenvalues are equal. 
In the following, we show that a Gaussian model with \textit{repeated eigenvalues}, i.e. isotropic in some multidimensional eigenspaces, has less parameters than one with \textit{distinct eigenvalues} and therefore provides a simpler explanation of the data. Then, using parsimonious model selection criteria such as the BIC, we are able to decide which eigenvalues should be assumed equal.

\section{Identifying the curse of isotropy}
In order to spot the curse of isotropy, we go through the lens of statistical modeling and introduce the PSA generative model. This model assumes a Gaussian distribution with repeated covariance eigenvalues. 
It enjoys an explicit maximum likelihood estimate with a rich geometry enabling effective model selection.

\subsection{PSA model}
Let ${\gamma} \coloneqq (\gamma_1, \dots, \gamma_d)$ be a \textit{composition} of a positive integer $p$---i.e. a sequence of positive integers that sums up to $p$.
We define the PSA model of \emph{type} ${\gamma}$ as the family of Gaussian distributions $~{p(x | \mu, \Sigma) \coloneqq \mathcal{N}(x | \mu, \Sigma)}$, where $~{\mu \in \R^p}$ is a mean vector and $~{\Sigma = \sum_{k=1}^d \lambda_k Q_k {Q_k}\T\in S_p^{++}}$ is a covariance matrix with repeated eigenvalues $\lambda_1 > \dots > \lambda_d > 0$ of respective multiplicity $\gamma_1, \dots, \gamma_d$ and associated eigenspaces $\mathrm{Im}(Q_1), \dots, \mathrm{Im}(Q_d)$.
These distributions can be rewritten as a (linear-Gaussian) latent variable generative model
\begin{equation}\label{eq:PSA_model}
{x} = \sum_{k=1}^{d-1} \sigma_k {Q}_k {z}_k + {\mu} + {\epsilon},
\end{equation}
where $~{\sigma_1 > \dots > \sigma_{d-1} > 0}$ are decreasing scaling factors,
${Q}_k \in \R^{p \times \gamma_k}$ are mutually-orthogonal $\gamma_k$-frames, $~{{z}_k \sim \N{{0}, {I}_{\gamma_k}}}$ are independent latent variables and ${\epsilon} \sim \N{{0}, \sigma^2 {I}_{p}}$ is an isotropic Gaussian noise. 
An illustration of the generative model is provided in Fig~\ref{fig:PSA}. 
PPCA and IPPCA models can then be reinterpreted as PSA models, of respective types $~{{\gamma} = (1, \dots, 1, p - q)}$ and ${\gamma} = (q, p - q)$, where $q < p$ is the intrinsic dimension (cf. Section~\ref{appsec:PSA}).
\begin{figure}
\centering
\includegraphics[width=\linewidth]{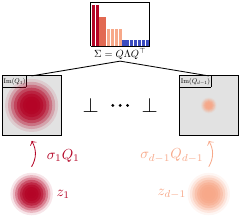}
\caption{PSA generative model, assuming that the observed data was first sampled from a sequence of independent low-dimensional normal latent variables, then linearly mapped to mutually-orthogonal subspaces and finally shifted and added an isotropic Gaussian noise~\eqref{eq:PSA_model}. The resulting density is a multivariate Gaussian with repeated eigenvalues, of respective multiplicities ${\gamma} = (2, 2, 4, 7)$.}
\label{fig:PSA}
\end{figure}

\subsection{Geometry and inference}
From a geometric point of view, the fitted density is isotropic on a sequence of mutually-orthogonal subspaces $\operatorname{Im}(Q_1) \perp \dots \perp \operatorname{Im}(Q_{d})$ of respective dimensions $\gamma_1, \dots, \gamma_d$.
Such a sequence is called a \emph{flag} of linear subspaces of \emph{type} ${\gamma}$.
Therefore, flags of type $\gamma$---which are diffeomorphic to $\O(p) / (\O(\gamma_1) \times \dots \times \O(\gamma_d))$~\cite{arnold_modes_1972, ye_optimization_2022}---naturally parameterize PSA models. 
Consequently, Stiefel manifolds and Grassmannians---which are particular cases of flag manifolds---respectively parameterize PPCA and IPPCA models (cf. Section~\ref{appsec:PSA}).
The remaining model parameters are the subspace variances $(\lambda_1, \dots, \lambda_d) \in \R^{d}$ and the mean ${\mu} \in \R^p$.
Thus, the \textit{complexity} (dimension of the parameter space) of the PSA model of type ${\gamma}$ is
\begin{equation}\label{eq:kappa}
    \kappa({\gamma}) \coloneqq p + d + \frac{p(p-1)}{2} - \sum_{k=1}^{d} \frac {\gamma_k (\gamma_k - 1)} {2}.
\end{equation}
We can notably see that the decrease in model complexity is quadratic in the number of equalized eigenvalues.

One of the strength of the PSA models is that their maximum likelihood estimate is \textit{explicit}, similarly to PPCA and IPPCA. In a nutshell, we show in Theorem~\ref{appthm:PSA} that the most likely mean vector {$\hat\mu$} is the \textit{empirical mean}, the most likely {eigenvalues $\hat\lambda_1, \dots, \hat\lambda_d$} are the \textit{block-averaged sample eigenvalues} according to the type $\gamma$, and the most likely flag {$(\operatorname{Im}(\hat Q_1), \dots, \operatorname{Im}(\hat Q_{d}))$} is the sequence of mutually-orthogonal subspaces spanned by the associated eigenvectors. {Denoting $\ell_1\geq\dots\geq\ell_p$ the sample eigenvalues, $q_k \coloneqq \sum_{l=1}^k \gamma_k$ the accumulated dimensions, and $~{\hat{\lambda}_k \coloneqq \frac{1}{\gamma_k}\sum_{j=q_{k-1}+1}^{q_{k}} \ell_j}$, the block-averaged sample eigenvalues, we get} the following expression for the maximum likelihood
\begin{equation}\label{eq:PSA_ML}
    \ln \hat{\mathcal{L}} (\gamma) = -\frac n 2 \left(p \ln(2\pi) + \sum_{k=1}^d \gamma_k \ln{{\hat{\lambda}_k}} + p\right).
\end{equation}

\subsection{Identifying the curse of isotropy in practice}
The Bayesian information criterion~\cite{schwarz_estimating_1978} is defined as 
\begin{equation}\label{eq:BIC}
    \operatorname{BIC} (\gamma) \coloneqq \kappa (\gamma) \ln n - 2 \ln \hat{\mathcal{L}} (\gamma).
\end{equation}
It is a widely-used model selection criterion, making a tradeoff between model complexity and goodness-of-fit, to prevent from overfitting given the number of observed samples. The formula results from an asymptotic approximation of the Bayesian model evidence. Given a dataset, one can compare the BIC of a PSA model with repeated eigenvalues to the BIC of a PSA model with distinct eigenvalues. The model with the lowest BIC is selected over the other one.

As discussed previously, two adjacent sample eigenvalues with a relatively small gap may be prone to isotropic PC variability. 
To identify such situations where the curse of isotropy may arise, we compare a \textit{full} covariance model $\gamma = (1, \dots, 1)$ with an \textit{equalized} covariance model $~{\gamma' = (1, \dots, 1, 2, 1, \dots, 1)}$ where eigenvalues $j$ and $j+1$ are assumed equal.
Denoting {$\delta(\ell_j, \ell_{j+1})$} $\coloneqq \frac{\ell_{j} - \ell_{j+1}}{\ell_j}$ the \textit{relative eigengap} between the two sample eigenvalues, we show in Theorem~\ref{appthm:releigengap} that

{
\begin{equation}\label{eq:releigengap_threshold}
    \operatorname{BIC}(\gamma') < \operatorname{BIC}(\gamma) \iff \delta(\ell_j, \ell_{j+1}) < \delta^{\mathrm{BIC}}(n),
\end{equation}
with $\delta^{\mathrm{BIC}}(n) = 2 (1 - n^{\frac2n} + n^{\frac1n}\sqrt{n^{\frac2n} - 1})$.
}

This condition---independent of $p$---is illustrated in Fig~\ref{fig:releigengap_curves} {(dark blue)}.
\begin{figure}
    \centering
    \includegraphics[width=\linewidth]{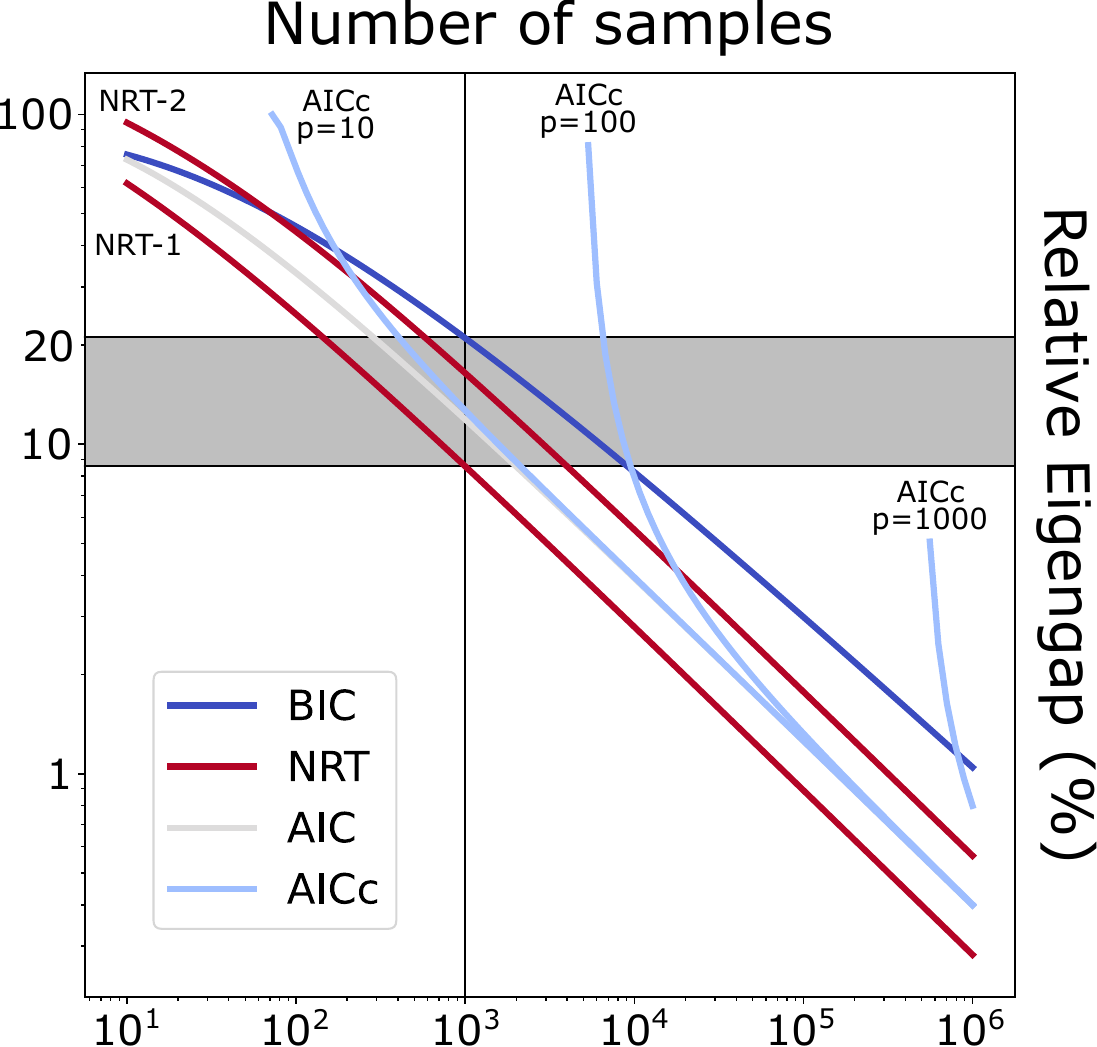}
    \caption{{Plot of the relative eigengap thresholds---under which two adjacent sample eigenvalues should be assumed equal---as a function of $n$, for different model selection criteria: the Bayesian information criterion (BIC)~\cite{schwarz_estimating_1978}, the Akaike information criterion (AIC)~\cite{akaike_new_1974}, its small-sample version (AICc)~\cite{hurvich_regression_1989}, and North's rule-of-thumbs (NRT)~\cite{north_sampling_1982}, which are all thoroughly worked out in Section~\ref{appsec:MS}. 
    For $n=1000$, all the methods have a relative eigengap threshold roughly between $10\%$ and $20\%$, which substantiates the importance of the curse of isotropy, whatever the chosen methodology.}}
	\label{fig:releigengap_curves}
\end{figure}
We notably deduce by substitution that for $n = 1000$ samples, all the adjacent sample eigenvalues with a relative eigengap lower than $\delta = 21\%$ should be assumed equal. In other words, given two sample eigenvalues of respective magnitude $1$ and $0.8$, one needs \textit{at least} $1000$ samples to overcome the curse of isotropy. \textit{This is rarely the case in practice.} To illustrate this, we test the condition~\eqref{eq:releigengap_threshold} on many classical datasets from the UCI Machine Learning Repository (cf. Section~\ref{appsec:data}), with $n/p$ ratios ranging from $10$ to $10^4$.
For each dataset, we report the pairs of adjacent eigenvalues that are below the relative eigengap threshold in Fig~\ref{fig:releigengap_UCI}.
\begin{figure}
\centering
\includegraphics[width=\linewidth]{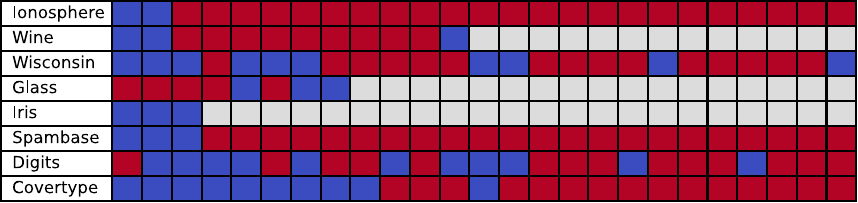}
\caption{Practical identification of the curse of isotropy on several classical datasets from the UCI Machine Learning Repository. A red case in column $j$ indicates that eigenvalues $j$ and $j+1$ have a relative eigengap below the threshold~\eqref{eq:releigengap_threshold} and should be equalized. Blue is above and gray is undefined (we only plot the $25$ leading eigenvalue pairs). We can clearly see that the curse of isotropy is not a negligible phenomenon in practice.}
\label{fig:releigengap_UCI}
\end{figure}
The outcomes are striking: all datasets but one have some eigenvalue pairs below the threshold. This does not only concern the smallest eigenvalues---which are usually tossed away because considered as noise---but also the {largest} ones---which are usually interpreted by applied scientists.
This shows that the curse of isotropy is not a negligible phenomenon at all and that particular care should be taken before interpreting the principal components.
Note that~\eqref{eq:releigengap_threshold} involves the \textit{relative} eigengap between adjacent eigenvalues and not the \textit{absolute} one, meaning that an exponentially-decreasing sample eigenvalue profile can actually highly suffer from the curse of isotropy. In other words, PSA models are not just suited to piecewise-constant-like sample covariance profiles.

{The power of the relative eigengap---seen as a test statistic to identify the curse of isotropy---is evaluated in Section~\ref{appsec:evaluation}. 
The condition~\eqref{eq:releigengap_threshold} tends to equalize more eigenvalues than necessary when the population relative eigengap (\textit{effect size}) and number of samples (\textit{sample size}) are small. But interestingly, this (too parsimonious) model misspecification tends to not only reduce the \textit{variance} of the underlying estimator, but also its \textit{bias}. The interest of PSA models therefore goes beyond the assumption that the true covariance matrix has multiple eigenvalues. 
}

{To go beyond the BIC, which} is known for its tendency to select underparameterized models~\cite{bishop_pattern_2006}, we also investigate in Section~\ref{appsec:MS} the eigenvalue-equalization guideline under other model selection criteria {such as} the Akaike information criterion (AIC)~\cite{akaike_new_1974} and {sampling error-based approaches} {(North's rule-of-thumbs, NRT)}~\cite{north_sampling_1982}. We get relative eigengaps around $10-20\%$ for $n=1000$, and experimental results substantiating the curse of isotropy's importance.

\subsection{Stratification and efficient model selection}
We now explicit the stratified structure of PSA models and show how it enables to design efficient model selection strategies to choose which groups of eigenvalues to equalize. More details are given in Section~\ref{appsec:MS}.

The space of symmetric matrices can be stratified according to the sequence of eigenvalue multiplicities \cite{arnold_modes_1972,groisser_geometric_2017,breiding_geometry_2018}. This implies that the PSA models in dimension $p$ form a stratified exponential family~\cite{geiger_stratified_2001} of cardinal $2^{p-1}$, partially-ordered~\cite{taeb_model_2024} by the stratum-inclusion relation.
We illustrate the family of PSA models in Fig~\ref{fig:hasse_complexity}.
\begin{figure}
\centering
\includegraphics[width=\linewidth]{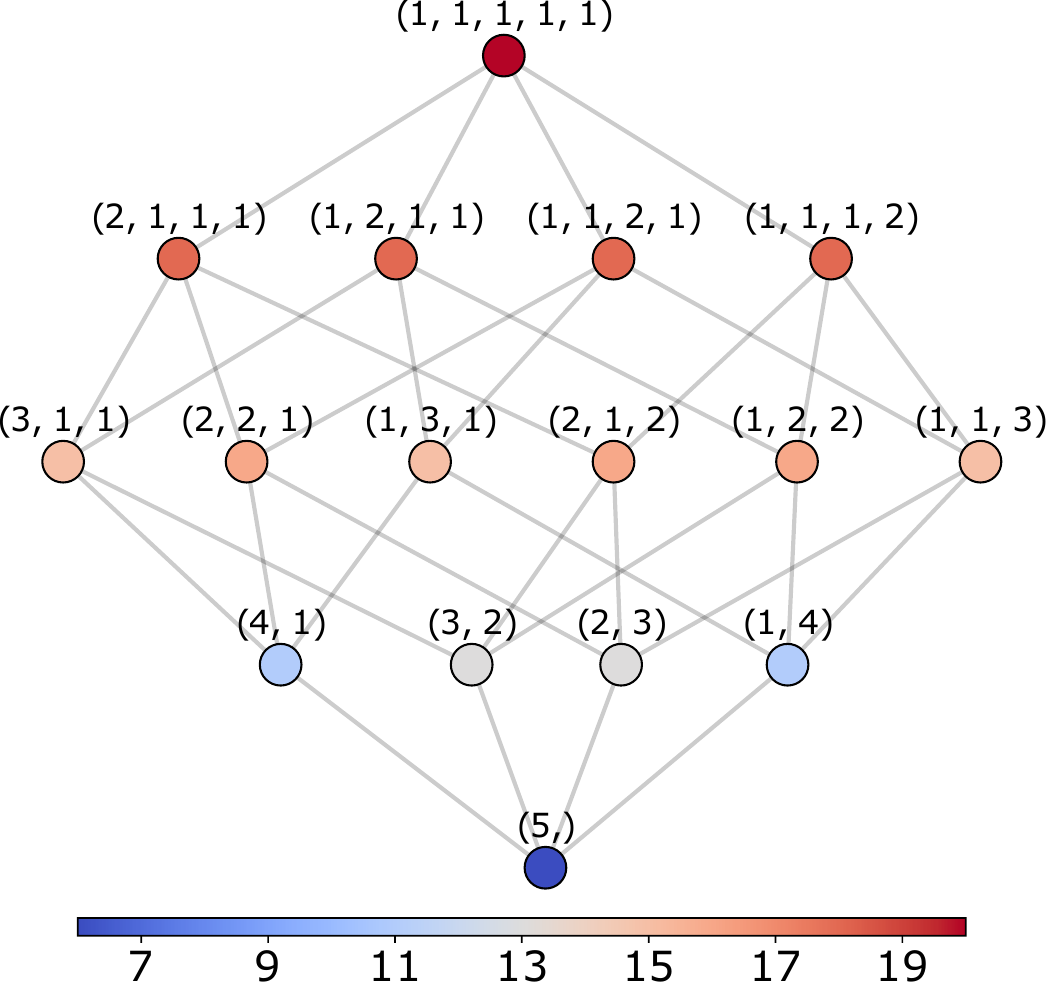}
\caption{Hasse diagram of $5$-dimensional PSA models. Each node represents a model. The associated label and color represent respectively the model type and its number of free parameters. The family contains $16$ models: the isotropic Gaussian is the bottom node, the full covariance model is the top node, the five PPCA models are on the right side and the four IPPCA models are located on the second level.}
\label{fig:hasse_complexity}
\end{figure}

In order to prevent from greedily exploring the whole family for model selection, we propose a simple yet efficient model selection technique based on the stratified structure of this family.
The \textit{hierarchical clustering strategy} consists in performing a hierarchical clustering of the sample eigenvalues, based on chosen \textit{pairwise distance} (e.g. the relative eigengap {$\delta(\ell_j, \ell_{j+1}) = (\ell_{j} - \ell_{j+1}) / \ell_j$}) and \textit{cluster-linkage criterion} (e.g. single-linkage {$\Delta(\Lambda_1, \Lambda_2) = \min_{\ell_1, \ell_2 \in \Lambda_1 \times \Lambda_2} \delta(\ell_1, \ell_2)$}). This strategy{, summarized in Algorithm~\ref{alg:hierarchical}, } yields a hierarchical subfamily of $p$ models with decreasing complexity, from which we can more efficiently select the model minimizing the BIC.
\begin{algorithm}[b]
   \caption{{Hierarchical clustering strategy for PSA}}
   \label{alg:hierarchical}
\begin{algorithmic}
   \State {\bfseries Input:} $\ell_1 \geq \dots \geq \ell_p, \Delta$ \Comment{sample eigenvalues and distance}
   \State ${\gamma}^{1} \gets \lrp{1, \dots, 1}, \quad {{\Lambda}}^{1} \gets \lrp{\lrs{\ell_1}, \dots, \lrs{\ell_p}}$ \Comment{full cov. init.}
   \For{$t = 1 \dots p-1$}
    \State $\Delta^{t} \gets (\Delta({\Lambda}^{t}_1, {\Lambda}^{t}_{2}), \dots, \Delta({\Lambda}^{t}_{p-t}, {\Lambda}^{t}_{p-t+1}))$ 
    \State $k^{t} \gets \argmin \Delta^{t}$ 
    \State ${{\Lambda}}^{t+1} \gets \lrp{{\Lambda}^t_1, \dots, {\Lambda}^t_{{k^t}-1}, {\Lambda}^t_{k^t} \cup {\Lambda}^t_{{k^t}+1}, {\Lambda}^t_{{k^t}+2}, \dots, {\Lambda}^t_d}$ 
    \State ${\gamma}^{t+1} \gets \lrp{{\gamma}^t_1, \dots, {\gamma}^t_{{k^t}-1}, {\gamma}^t_{k^t} + {\gamma}^t_{{k^t}+1}, {\gamma}^t_{{k^t}+2}, \dots, {\gamma}^t_d}$ 
\EndFor
    \State $\hat\gamma \gets 
    \argmin_{\gamma\in\lrp{{\gamma}^{t}}_{t=1}^p} \operatorname{BIC}(\gamma)$
    \State {\bfseries Output:} $\hat\gamma$ \Comment{selected eigenvalues multiplicities}
\end{algorithmic}
\end{algorithm}
We prove the \textit{asymptotic consistency} of the hierarchical clustering strategy in Proposition~\ref{appprop:hierarchical_heuristic}, as well as introduce other strategies. 
{We evaluate the model selection accuracy of the hierarchical clustering strategy in Section~\ref{appsec:evaluation}. We get a sharp transition between the ``small $n$ small $\delta$'' and the ``large $n$ large $\delta$'' regimes, where the accuracy goes from $0$ to $100\%$.}

\section{From principal components to principal subspaces}

To summarize the previous section, parsimonious considerations invite us to block-average eigenvalues whose relative gaps are close---given the number of observed samples. The associated PSA model is now parameterized with \textit{eigenspaces} instead of individual \textit{eigenvectors} and we are therefore facing the curse of isotropy.
In this section, we propose {a few ideas} to improve data interpretability {in this context,} by transitioning from \textit{principal components} to \textit{principal subspaces}.

A first idea, rather \textit{quantitative}, is to look for rotations of principal components inside {the principal subspace they span} in order to increase {an interpretability-related criterion $f$}{:
\begin{equation}
    Q'_k = Q_k \,\argmax_{R_k\in\O(\gamma_k)} f(Q_k R_k).
\end{equation}
}
Indeed, as explained previously, the curse of isotropy might cause principal components to be rotated versions of more interpretable {latent variables}.
\textit{Varimax} rotation~\cite{kaiser_varimax_1958,rohe_vintage_2023} enables for instance to get rotated components with sparse loadings.
Many other criteria {$f$} can be considered depending on the data type {and the objective. For instance, if the data are images, then one can use local entropy, structured sparsity~\cite{jenatton_structured_2010} or total variation criteria to get sharp components}.
{The orthogonal transformations $R_k\in\O(\gamma_k)$ can also be replaced with more general linear transformations $A_k\in \R^{\gamma_k\times \gamma_k}$ if one does not need orthogonal components. An interesting idea in that sense is to perform an independent component analysis (ICA)~\cite{hyvarinen_independent_2000} inside each principal subspace. Indeed, under the PSA model, the projected distribution is isotropic Gaussian, but under another model (e.g. Laplacian), it might have privileged directions.  This ``PSA+ICA'' idea interestingly provides independent components with a hierarchy (related to the explained variance) while independent components are usually unordered.
}

A second idea, rather \textit{qualitative} {and exploratory}, is to generate samples from the multidimensional principal subspaces via Eq.~\eqref{eq:PSA_model} (cf. Fig~\ref{fig:PSA}) {and inspect them visually}. Those samples might have common characteristics like {similar} frequencies or {other} invariances~\cite{hyvarinen_emergence_2000}. {Instead of generating Gaussian samples from the principal subspaces}, one can {generate uniform samples on an \textit{inscribed sphere}} to {explore the principal subspaces more exhaustively}. 
{Finally, if one has an intuition about how the variability modes should look like (as it can be the case in climate science~\cite[Section~4.3]{jolliffe_principal_2002} for instance), or if one possesses interpretable co-variables (e.g. the age associated with a patient's image), then one can use these extra features to enhance the interpretation of the principal subspaces.
}

\section{Experiments}
Eventually, PSA is grounded in a generative model with a rich geometry, yet the methodology is very simple and can be summarized in the three following steps:
\textit{eigendecomposition} of the sample covariance matrix, \textit{block-averaging} of the eigenvalues with small relative eigengaps (or more formally, PSA model selection), \textit{interpretation} of the resulting principal subspaces via factor rotation or latent subspace sampling. 
In this section we apply the PSA methodology to several synthetic and real datasets in a variety of fields. The experiments show that principal components associated with relatively-close eigenvalues are generally {fuzzy} due to the curse of isotropy. Therefore, equalizing the problematic eigenvalues and lifting the analysis to principal subspaces dramatically enhances exploratory data analysis.

\subsection{Laplacian eigenfunctions}
In this experiment, we generate a synthetic dataset consisting in linear combinations of Laplacian eigenfunctions (also known as \textit{quasimodes}~\cite{arnold_modes_1972}) with variance being a decreasing function of the Laplacian eigenvalue. This kind of generative model has been extensively used in many different areas, notably climate sciences~\cite{north_sampling_1982} (for modeling atmospheric fields on the earth) and computer vision (for modeling shadows on faces under varying illumination conditions~\cite{basri_lambertian_2003} or low-frequency patches in natural images~\cite{field_relations_1987}). The global idea behind those models is that natural symmetries are present in the shapes under study (face, earth, square domain etc.) and lead to multiple eigenvalues in their Laplacian matrix, and therefore to multiple eigenvalues in the covariance matrix of homogeneous stochastic processes on those shapes.

We generate $n=600$ points on a square grid with 64 pixels on each side. We take a combination of the 	$q=9$ leading eigenmodes with variance scaling like $\exp(-\lambda)$ (where $\lambda$ is the Laplacian eigenvalue) and add an isotropic noise.
We fit a PSA model of type $\gamma' = (1, 2, 1, 2, 2, 1, 4087)$ (corresponding to the expected Laplacian eigenvalue multiplicities on a square domain) and compare it to the associated PPCA model $\gamma = (1, \dots, 1, 4087)$. We get a lower BIC for the PSA model, then perform {ICA} in{side} each eigenspace, and finally uniformly sample from the unit sphere inscribed in the 2D principal subspaces. The results are shown in Fig~\ref{fig:synthetic_exp}.
\begin{figure}
\centering
\includegraphics[width=\linewidth]{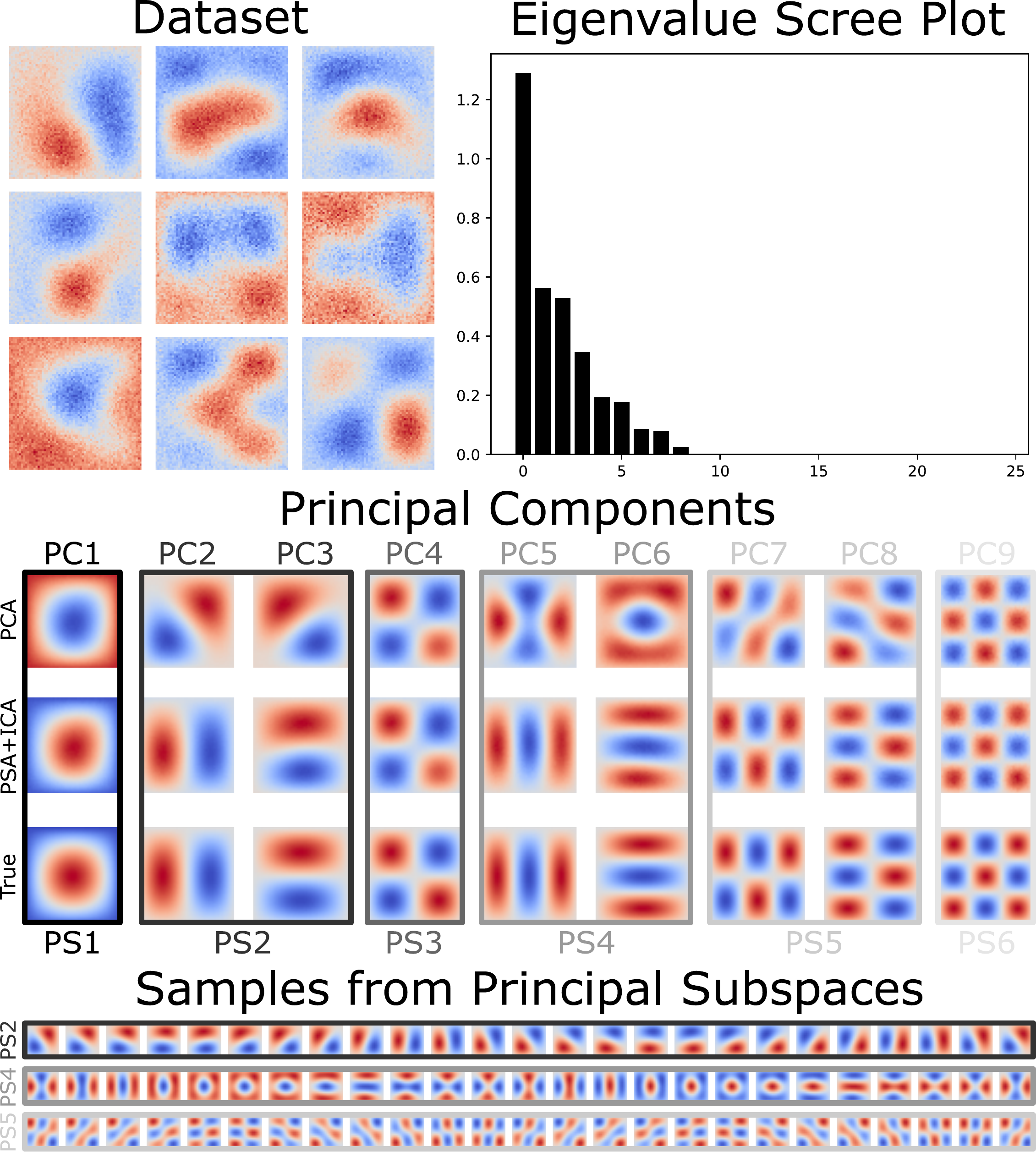}
\caption{PCA vs PSA on the Laplacian eigenfunction dataset. 
Top: Dataset and covariance eigenvalue scree plot.
Middle: Principal components (observed, rotated and true). The observed principal components are linear combinations of the true eigenmodes---i.e. quasimodes---especially $(5,6)$ and $(7, 8)$. 
After {ICA}, one recovers the original eigenmodes.
Bottom: Principal subspaces resulting from a PSA model of type $\gamma' = (1, 2, 1, 2, 2, 1, 4087)$. We sample from those 2D subspaces and obtain equal-frequency quasimodes.
}
\label{fig:synthetic_exp}
\end{figure}
We can see that the principal components are linear combinations of the original eigenmodes, i.e. quasimodes. With {ICA} in{side} the associated subspaces, we better recover the original modes. {Moreover, the principal subspaces are effectively ordered according to their intrinsic frequencies (which can be measured by the number of ``stripes'' in the images) and the equal-frequency quasimodes are gathered in the same subspaces.}

\subsection{Natural image patches}
In this experiment, we consider patches extracted from natural images, as done in many seminal works investigating biological vision via unsupervised machine learning methods~\cite{field_relations_1987,olshausen_emergence_1996, hyvarinen_emergence_2000,mairal_online_2010}.
We consider 10 flower images from the Natural Images database (cf. Section~\ref{appsec:data}) and randomly extract $n=500$ $(8, 8)$-pixel patches from those. After removing the DC component (i.e. the mean value) to each patch, like usually done in such studies~\cite{hyvarinen_emergence_2000} and looking at the sample eigenvalue profile, we decide to fit a PSA model of type $\gamma' = (2, 3, {4, 55})$ and compare it to the associated PPCA model $\gamma = (1, \dots, 1, {55})$. We get a lower BIC with the PSA model. Then, we uniformly sample from the unit {sphere inscribed in the first (2D) principal subspace and randomly (Gaussian) sample from the second (3D) and third (4D) principal subspaces}. We report the results in Fig~\ref{fig:exp_patch}.
\begin{figure}[t]
\centering
\includegraphics[width=\linewidth]{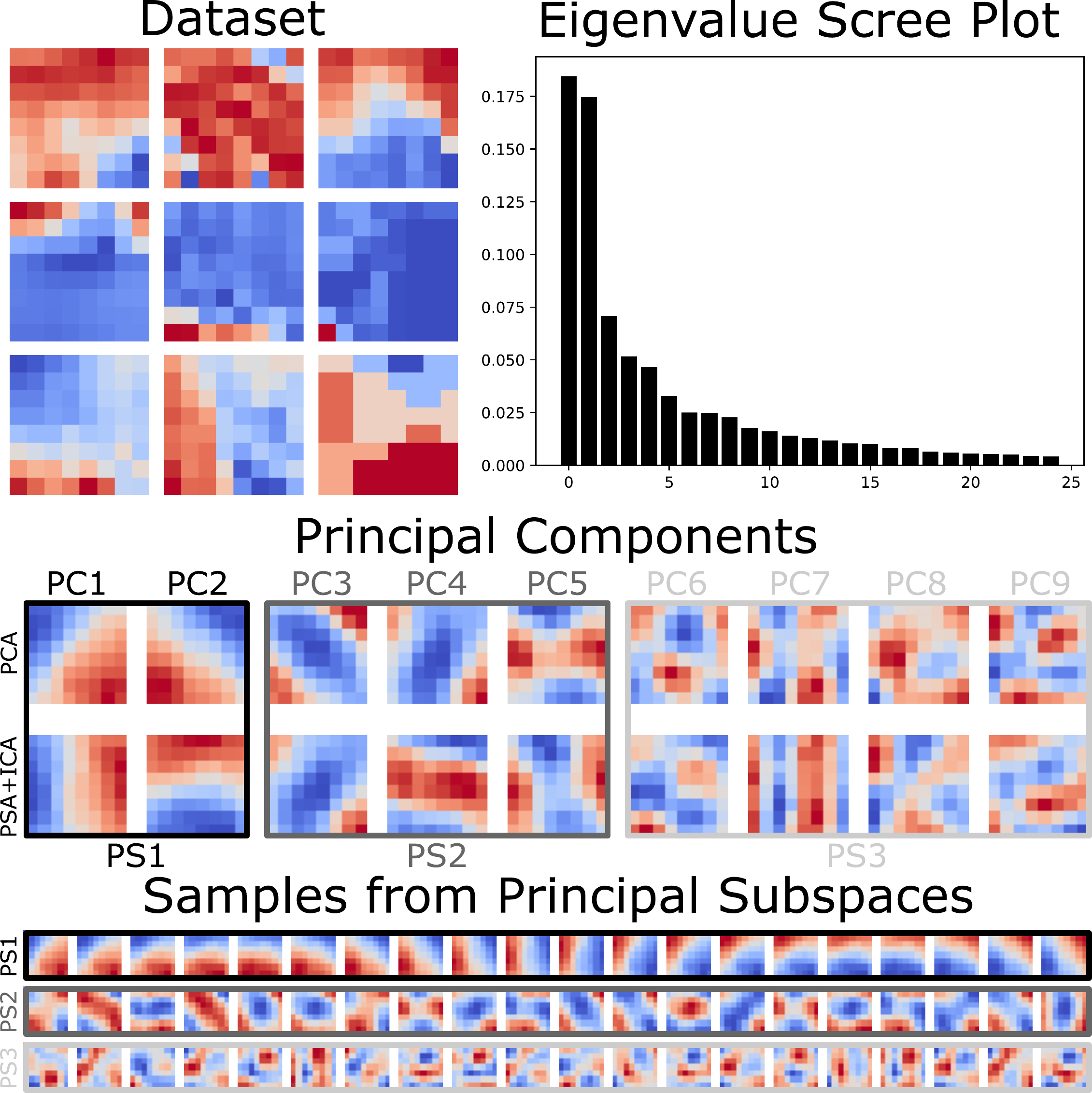}
\caption{PCA vs PSA on the natural image patch dataset. 
Top: Dataset and covariance eigenvalue scree plot.
Middle: Principal components.
Bottom: Principal subspaces resulting from a PSA model of type $\gamma'=(2, 3, 4, 55)$. We sample from those 2D and 3D subspaces and notice the emergence of {decreasing}-frequency feature subspaces with {(limited) invariances}. 
We insist on the fact that without principal subspace analysis, we would not-necessarily have been able to detect those multidimensional patterns.}
\label{fig:exp_patch}
\end{figure}
While principal components do not look particularly interpretable individually, grouping them into principal subspaces with isotropic variability brings out low-frequency subspaces with {(limited) invariances~\cite{hyvarinen_emergence_2000}}.
From a curse-of-isotropy point of view, the observed principal components are random samples from the illustrated {principal subspaces}.

\subsection{Eigenfaces}
In this experiment, we consider the Carnegie Mellon University (CMU) Face Image database (cf. Section~\ref{appsec:data}). It consists in 640 grayscale pictures of people from varying pose, expression, and eye conditions. We extract $n=31$ $(60, 64)$-pixel images of the subject \textit{Choon}. Inspired by the seminal paper~\cite{basri_lambertian_2003} (which establishes a link between face shadowing under varying illumination conditions and spherical harmonics), we fit a PSA model of type $\gamma' = (1, 3, 5, 3831)$ and compare it to the associated PPCA model $\gamma = (1, \dots, 1, 3831)$. We get a lower BIC and perform {ICA} in the second {principal} subspace, which is 3-dimensional. The results are illustrated in Fig~\ref{fig:exp_CMU}.
\begin{figure}
\centering
\includegraphics[width=\linewidth]{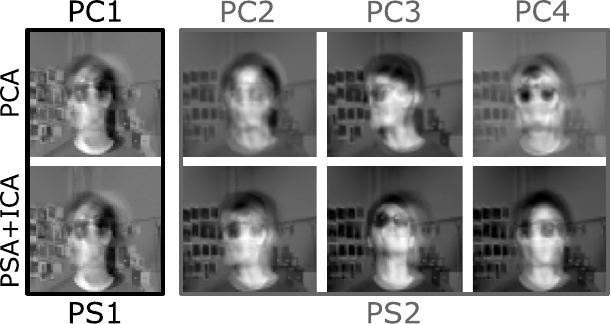}
\caption{PCA vs PSA on the CMU Face Image dataset.
Top: Eigenfaces.
Bottom: Eigenfaces after {ICA} within the second principal subspace (spanned by components 2, 3 and 4).
This experiment shows that the curse of isotropy can yield blurry eigenfaces, corresponding to linear combinations of much more interpretable components that we can recover with PSA{+ICA}.
}
\label{fig:exp_CMU}
\end{figure}
While principal components 2, 3, and 4 are fuzzy and uninterpretable, we can see that they actually correspond to linear combinations of three \textit{much more interpretable} factors{, related to head movements}. 

{Another possible approach that we tried is to generate sample images from the second principal subspace, ordered according to their local entropy, and then select among the samples with the lowest entropy the most visually-insightful ones. We also recovered head movements that are slightly sharper than the ones in Fig~\ref{fig:exp_CMU}.}

\subsection{Structured data}
In this experiment, we consider a structured dataset taken from the UCI ML repository. The \textit{Glass identification} dataset (cf. Section~\ref{appsec:data}) from the USA Forensic Science Service contains chemical features about different types of glasses, with applications in criminology. We fit a PSA model of type $\gamma' = (5, 4)$ and compare it to the associated PPCA model $\gamma = (1, \dots, 1, 4)$. We get a lower BIC and perform varimax rotation in the first feature subspace, of dimension 5. We report the loadings of the sample eigenvectors and compare them to the PSA factors after rotation in Fig~\ref{fig:exp_rotation}.
We see that the PSA factors are more interpretable than the principal components, in the sense that they express as sparser combinations of the original variables. Moreover, contrary to classical factor rotation methods done after PCA (cf. Chapter~11 of \cite{jolliffe_principal_2002}), we here do not lose any hierarchy in the principal components in terms of explained variance, since under the PSA model of type $\gamma' = (5, 4)$, the five components have equal variance.
\begin{figure}
\centering
\includegraphics[width=\linewidth]{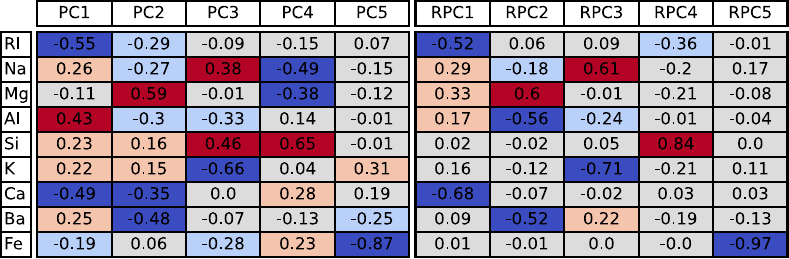}
\caption{PCA vs PSA on the Glass identification dataset. We fit a PSA model of type $\gamma' = (5, 4)$ and perform varimax rotation in the first principal subspace. We can see that the rotated PCs (right) are much sparser in the original variables than the PCs (left), while having the same estimated variance (under the PSA model). 
The colors are stratified similarly as in Section~4.1 of \cite{jolliffe_principal_2002} to help interpretability. For each eigenvector, the cases in dark red and dark blue correspond to coefficients whose absolute value is greater than half of the maximal absolute coefficient, the ones with light red and light blue to coefficients whose absolute value is between one quarter and one half, and the ones in gray are below, considered as negligible.}
\label{fig:exp_rotation}
\end{figure}

\section{Related works}\label{sec:related}
In the climate research community, a celebrated work often cited as \textit{North's rule-of-thumb}~\cite{north_sampling_1982} warns scientists against close eigenvalues in the Karhunen-Loève expansion of a meteorological field. Indeed, the associated principal components---referred to as \textit{empirical orthogonal functions} (EOF)---suffer from large sampling errors, which is very problematic due to the key role EOF's play in this field for exploratory data analysis. The authors provide a perturbation-theoretical rule-of-thumb to decide which eigenvalues form \textit{degenerate multiplets}. The rule as stated in the paper is quite vague, however we are able in Section~\ref{appsec:MS} to reformulate its practical software implementation as a relative eigengap threshold and to compare it to our criterion~\eqref{eq:releigengap_threshold}. We show that this threshold is much lower than ours (e.g. $8.6\%$ instead of $21\%$ for $1000$ samples), therefore our result has a much larger impact on the practical methodology of PCA.

More broadly, several works have mentioned close-eigenvalues in PCA or in general symmetric matrices. 
A paper from Jolliffe~\cite{jolliffe_rotation_1989} shows the advantages of factor rotation inside subspaces spanned by principal components with close eigenvalues for {tabular} data. 
{Permuting eigenvectors with similar eigenvalues is commonplace in spectral shape analysis~\cite[Section~2.3]{lombaert_focusr_2013}.} 
Eigenvalue equality has also been studied formally in the context of oscillatory systems~\cite{arnold_modes_1972,lazutkin_kam_1993,gershkovich_problem_2004} diffusion tensor imaging~\cite{groisser_geometric_2017}, spectral geometry~\cite{besson_multiplicy_1988}, statistical tests~\cite{anderson_asymptotic_1963,tyler_asymptotic_1981,schwartzman_inference_2008,rabenoro_geometric_2024} etc.

Finally, the use of flags for statistical analysis has been particularly well illustrated with the example of \textit{independent subspace analysis}~\cite{hyvarinen_emergence_2000}, from which the name of our model is drawn. The authors notice the emergence of phase and shift-invariant features by maximizing the independence between the norms of projections of samples into so-called \textit{independent feature subspaces}. The learning algorithm is later recast as an optimization problem on flag manifolds~\cite{nishimori_riemannian_2006}. 
Flags also implicitly arise in general subspace methods under the name \textit{mutually orthogonal subspaces}, like in the mutually-orthogonal class-subspaces of Watanabe and Pakvasa~\cite{watanabe_subspace_1973} and the adaptive-subspace self-organizing maps of Kohonen~\cite{kohonen_emergence_1996}.
More recently, PCA was also reformulated as an optimization problem on flag manifolds~\cite{pennec_barycentric_2018}, raising perspectives for multilevel data analysis on manifolds.

\section{Discussion}
We raised an important issue---the \textit{curse of isotropy}---about the isotropic variability of principal components under Gaussian models with repeated covariance eigenvalues, and showed that these models should often be assumed in practice according to the principle of parsimony. We developed a simple methodology---\textit{principal subspace analysis}---based on generative modeling and flags of subspaces to spot this curse in practice and transition from fuzzy principal components to much-more-interpretable principal subspaces.

Principal subspace analysis paves the way to numerous extensions.
First, one could deal with non-Gaussian data ({elliptical distributions}, Gaussians on manifolds~\cite{pennec_intrinsic_2006}, {Lie group orbits~\cite{ennes_liedetect_2025},} deep generative models etc.). 
In that case, the maximum likelihood estimate might not be explicit and one might require tools from optimization on Riemannian manifolds~\cite{absil_optimization_2009} (flag manifolds~\cite{ye_optimization_2022, zhu_practical_2024,szwagier_nested_2025}, symmetric positive-definite matrices~\cite{groisser_geometric_2017}{, etc.}) {or stratified spaces~\cite{leygonie_gradient_2023,olikier_first-order_2023}}.
Second, one should investigate {alternative approaches for grouping similar eigenvalues.}
{Some ideas---such as penalizing the likelihood with $\ell^1$-penalties on the eigengaps~\cite{szwagier_eigengap_2025}, bootstrap-based eigenvalue-eigenvector stability analysis and Bayesian frameworks~\cite{minka_automatic_2000}---are discussed in Section~\ref{appsec:alternatives}.} 
Third, since PSA models are nothing but parsimonious Gaussian models, one could simply extend them into parsimonious Gaussian \textit{mixture} models~{\cite{tipping_mixtures_1999,bouveyron_high-dimensional_2007,szwagier_parsimonious_2025}}. The eigenvalue-equalization principle could actually be applied to any problem relying on symmetric matrices, like 
variational {inference}~\cite{lambert_limited-memory_2023} or spectral geometry.
Notably, we think that the PSA methodology could extend to spectral graph theory and applications~\cite{ng_spectral_2001,belkin_laplacian_2003,lefevre_perturbation_2023}, where relatively-close Laplacian eigenvalues are common (related to shape symmetries) and might be especially problematic for spectral embedding and spectral matching~\cite{lombaert_focusr_2013}.
Fourth, any method relying on flags of subspaces~\cite{nishimori_riemannian_2006, ma_flag_2021, szwagier_rethinking_2023, mankovich_chordal_2023, mankovich_fun_2024,mankovich_flag_2025} could benefit from our framework to select an adapted flag type, whose choice has been canonical or heuristic up to now.

\begin{acks}[Acknowledgments]
The authors would like to thank the anonymous referees, the Associate Editor and the Editor for their constructive comments and ideas that improved the quality of this paper. 
\end{acks}

\begin{funding}
This work was supported by the ERC grant \#786854 G-Statistics from the European Research Council under the European Union’s Horizon 2020 research and innovation program and by the French government through the 3IA Côte d’Azur Investments ANR-23-IACL-0001 managed by the National Research Agency.
\end{funding}

\begin{supplement}
\stitle{Code}
\sdescription{The code is available on GitHub: \href{https://github.com/tomszwagier/principal-subspace-analysis}{tomszwagier/principal-subspace-analysis}.}
\end{supplement}
\begin{supplement}
\stitle{Appendix}
\sdescription{
The following self-contained appendix details the principal subspace analysis methodology. It includes proofs, data information and additional theoretical and practical results highlighting the importance of the curse of isotropy.}
\end{supplement}


\bibliographystyle{imsart-number} 
\bibliography{bibliography}       

\onecolumn
\appendix

\section{Reminders on probabilistic principal component analysis}
Principal component analysis (PCA) is a ubiquitous tool in statistics, which however lacks a probabilistic formulation.
Such a framework can indeed be useful in a variety of contexts like decision-making, generative modeling, missing data and model selection.
The Probabilistic PCA model of~\cite{tipping_probabilistic_1999} circumvents this issue, and we describe it in this section.

\subsection{Model}
Let $\lrp{{x}_i}_{i=1}^n$ be a $p$-dimensional dataset and $q \in \lrb{0, p-1}$ a lower dimension. In PPCA, the observed data is assumed to stem from a $q$-dimensional latent variable via a linear-Gaussian model

\begin{equation}\label{appeq:PPCA_model}
    {x} = {W} {z} + {\mu} + {\epsilon},
\end{equation}
with ${z} \sim \N{0, {I}_q}$, ${W} \in \R^{p \times q}$, ${\mu} \in \R^p$, ${\epsilon} \sim \N{0, \sigma^2 {I}_p}$ and $\sigma^2 > 0$. 

Through classical probability theory, one can show that the observed data is modeled as following a multivariate Gaussian distribution
\begin{equation}
    {x} \sim \N{{\mu}, {W} {W}\T + \sigma^2 {I}_p}.
\end{equation}
An analysis of the covariance matrix reveals that the distribution is actually anisotropic on the first $q$ dimensions and isotropic on the remaining $p - q$ ones. Hence there is an implicit constraint on the covariance model of the data, which is that the lowest $p - q$ eigenvalues are assumed to be all equal.

\subsection{Maximum likelihood}
The PPCA model parameters are the shift ${\mu}$, the linear map ${W}$ and the noise factor $\sigma^2$. Let some observed dataset $\lrp{{x}_i}_{i=1}^n$, $\overline {{x}} \coloneqq \frac 1 n \sum_{i=1}^n {x}_i$ its mean and $~{{S} \coloneqq \sum_{j=1}^p \ell_j {v}_j {{v}_j}\T}$ its sample covariance matrix, with its eigenvalues $\ell_1 \geq \dots \geq \ell_p \geq 0$ and associated eigenvectors ${v}_1 \perp \dots \perp {v}_p$. One can explicitly infer the parameters that are the most likely to have generated these data using maximum likelihood estimation.
It is shown in the original PPCA paper that the most likely shift is the empirical mean, the most likely linear map is the composition of a scaling by the $q$ {largest} eigenvalues ${L}_q\coloneqq\diag{\ell_1, \dots, \ell_q}$ (up to the noise) and an orthogonal transformation by the associated $q$ eigenvectors ${V}_q\coloneqq\lrb{{{v}}_1|\dots|{{v}}_q}$, and finally the most likely noise factor is the average of the $p - q$ discarded eigenvalues
\begin{equation}\label{appeq:PPCA_ML}
\hat{{\mu}} = \overline {{x}} , \hspace*{15mm}
\hat{{W}} = {V}_q \lrp{{L}_q - \hat{\sigma}^2 {I}_q}^{\frac 1 2} , \hspace*{15mm}
\hat{\sigma}^2 = \frac 1 {p - q} \sum_{j=q+1}^p \ell_j.
\end{equation}
One can then easily express the maximum log-likelihood
\begin{equation}
    \ln \hat{\mathcal{L}}(q) \coloneqq -\frac n 2 \lrp{p \ln(2\pi) + \sum_{j=1}^q \ln{\ell_j} + (p - q) \ln\lrp{\frac 1 {p - q} \sum_{j=q+1}^p \ell_j} + p}.
\end{equation}

\subsection{Parsimony and model selection}
The previously described PPCA is already a parsimonious statistical model. Indeed, it not only makes the assumption that the observed data follows a multivariate Gaussian distribution, which is the entropy-maximizing distribution at a fixed mean and covariance, but it also reduces the number of covariance parameters by constraining the last $p-q$ eigenvalues to be equal.
The covariance matrix $\Sigma \coloneqq {W} {W}\T + \sigma^2 {I}_p$ is parameterized by ${W} \in \R^{p\times q}$ and $\sigma^2$. It is shown in the original PPCA paper to have $\kappa(q) \coloneqq p q - \frac{q (q-1)}{2} + 1$ free parameters---the removal of $\frac{q (q-1)}{2}$ parameters being due to the rotational-invariance of the latent variable $z\in\R^q$. Although not evident at first sight with this expression of $\kappa$, we have a drop of complexity---with respect to the full covariance model which is of dimension $\frac {p(p+1)}{2}$---due to the equality constraint on the low eigenvalues, and the number of parameters decreases along with $q$.
As discussed in the next section, we can give an insightful geometric interpretation to the number of free parameters in the PPCA model using Stiefel manifolds.

For a given data dimension $p$, a PPCA model is indexed by its latent variable dimension $q \in \lrb{0, p-1}$. The process of model selection then consists in comparing different PPCA models and choosing the one that optimizes a criterion, like the Bayesian information criterion (BIC) or more PPCA-oriented ones like Bayesian PCA~\cite{bishop_bayesian_1998} or Minka's criterion~\cite{minka_automatic_2000}. They often rely on a tradeoff between goodness-of-fit (via maximum likelihood) and complexity (via the number of parameters), weighted by the number of samples.

\subsection{Isotropic PPCA}
Isotropic PPCA (IPPCA)~{\cite{bouveyron_intrinsic_2011}} is an even more constrained covariance model with only two distinct eigenvalues. For $a > b$ and ${U} \in \R^{p \times q}$ such that ${U}\T {U} = {I}_q$, one defines it as
\begin{equation}
    \Sigma \coloneqq \lrp{a - b} {U} {U}\T + b {I}_p.
\end{equation}
Such a parsimonious model is shown to be efficient in high-dimensional classification problems~\cite{bouveyron_high-dimensional_2007}.
The authors derive the maximum likelihood of such a model, which is highly related to the one of PPCA, where this time the $q$ first sample covariance eigenvalues are also averaged to fit the model. They also show that the maximum likelihood criterion alone is surprisingly asymptotically consistent for selecting the true intrinsic dimension under the assumptions of IPPCA.

\section{Principal subspace analysis}\label{appsec:PSA}  

Inspired by the complexity drop induced by the isotropy in the noise space in PPCA, we aim at investigating more general isotropy constraints on the full data space.
In this section, we introduce PSA, a covariance generative model with a general constraint on the sequence of eigenvalue multiplicities. PSA generalizes PPCA and IPPCA and unifies them in a new family of models parameterized by flag manifolds. Flag manifolds are themselves generalizations of Stiefel manifolds and Grassmannians, hence the link between PPCA, IPPCA and PSA that is detailed in this section.

\subsection{Model}
We recall that in combinatorics, a \emph{composition} of an integer $p$ is an ordered sequence of positive integers that sums up to $p$.
It has to be distinguished from
a \emph{partition} of an integer, which doesn't take into account the ordering of the parts. 

Let ${\gamma} \coloneqq (\gamma_1, \gamma_2, \dots, \gamma_d) \in \mathcal{C}(p)$ be a composition of a positive integer $p$.
We define the PSA model of \emph{type} ${\gamma}$ as
\begin{equation}\label{appeq:PSA_model}
{x} = \sum_{k=1}^{d-1} \sigma_k {Q}_k {z}_k + {\mu} + {\epsilon}.
\end{equation}
In this formula, $~{\sigma_1 > \dots > \sigma_{d-1} > 0}$ are decreasing scaling factors,
${Q}_k \in \R^{p \times \gamma_k}$ are mutually-orthogonal $\gamma_k$-frames (i.e. they verify ${{Q}_k}\T {Q}_{k'} = \delta_{k k'} {I}$ in Kronecker notation)
and $~{{z}_k \sim \N{{0}, {I}_{\gamma_k}}}$ are independent latent variables.
${\mu} \in \R^p$, $\sigma^2 > 0$ and $~{{\epsilon} \sim \N{{0}, \sigma^2 {I}_{p}}}$ are the classical shift, variance and isotropic noise present in PPCA.

Similarly as in PPCA, we can compute the population density
\begin{equation}
    {x} \sim \N{{\mu}, \hspace{1mm} \sum_{k=1}^{d-1} {\sigma_k}^2 {Q}_k {{Q}_k}\T + \sigma^2 {I}_p}.
\end{equation}
The expression of the covariance matrix $~{{\Sigma} \coloneqq \sum_k {\sigma_k}^2 {Q}_k {{Q}_k}\T + \sigma^2 {I}_p \in \R^{p\times p}}$ can be simplified by gathering all the orthonormal frames into one orthogonal matrix $~{{Q} \coloneqq \lrb{{Q}_1 | \dots | {Q}_{d-1} | {Q}_d} \in \O(p)}$ where ${Q}_d \in \R^{p \times \gamma_d}$ is an orthogonal completion of the previous frames. 
Writing ${\Lambda} \coloneqq \diag{\lambda_1 {I}_{\gamma_1}, \dots, \lambda_d {I}_{\gamma_{d}}}$, with 
$\lambda_k \coloneqq {\sigma_k}^2 + \sigma^2$ for $k\in \lrb{1, d-1}$ and $\lambda_d \coloneqq \sigma^2$, one gets
\begin{equation}\label{appeq:cov_simpl}
    {\Sigma} = {Q} {\Lambda} {Q}\T.
\end{equation}
Hence, the fitted density of PSA is a multivariate Gaussian with repeated eigenvalues $~{\lambda_1 > \dots > \lambda_d > 0}$ of respective multiplicity $~{\gamma_1, \dots, \gamma_d}$. 
An illustration of the generative model is provided in Fig~\ref{fig:PSA}.
Therefore, PPCA and IPPCA can be seen as PSA models, with respective types $~{{\gamma} = (1, \dots, 1, p - q)}$ and ${\gamma} = (q, p - q)$.
From a geometric point of view, the fitted density is isotropic on the  eigenspaces of ${\Sigma}$, which constitute 
a sequence of mutually-orthogonal subspaces of respective dimension $\gamma_1, \dots, \gamma_d$, whose direct sum generates the data space.
Such a sequence is called a \emph{flag} of linear subspaces of \emph{type} ${\gamma}$~{\cite{ye_optimization_2022}}.
Hence flags are natural objects to geometrically interpret PSA, and so a fortiori PPCA and IPPCA. We detail this point in the next section.

\subsection{Type}
Just like the latent variable dimension $q \in \lrb{0, p-1}$ is a central notion in PPCA, the type ${\gamma} \in \mathcal{C}(p)$ is a central notion in PSA. In this subsection, we introduce the concepts of \emph{refinement} and \emph{${\gamma}$-composition} to make its analysis more convenient.

Let ${\gamma} \coloneqq (\gamma_1, \gamma_2, \dots, \gamma_d) \in \mathcal{C}(p)$. We say that ${\gamma}' \in \mathcal{C}(p)$ is a \emph{refinement} of ${\gamma}$, and note ${\gamma} \preceq {\gamma}'$, if we can write ${\gamma}' \coloneqq ({\gamma}'_1, {\gamma}'_2, \dots, {\gamma}'_d)$, with ${\gamma}'_k \in \mathcal{C}(\gamma_k), \forall k \in \lrb{1, d}$. For instance, one has $(2, 3) \preceq (1, 1, 2, 1)$, while $(2, 3) \npreceq (3, 2)$ and $(3, 2) \npreceq (2, 3)$.

Let ${\gamma} \coloneqq (\gamma_1, \gamma_2, \dots, \gamma_d) \in \mathcal{C}(p)$. 
Then each integer between $1$ and $p$ can be uniquely assigned a \emph{part} of the composition, indexed between $1$ and $d$. 
We define the \emph{$\gamma$-composition function} $\phi_{{\gamma}}\colon \lrb{1, p} \to \lrb{1, d}$ to be this surjective map, such that $\phi_{{\gamma}}(j)$ is the index $k$ of the part the integer $j$ belongs to. 
For instance, one has $~{\phi_{(2, 3)}(1) = \phi_{(2, 3)}(2) = 1}$ and $~{\phi_{(2, 3)}(3) =  \phi_{(2, 3)}(4) = \phi_{(2, 3)}(5) = 2}$.
Then, intuitively and with slight abuse of notation, each object of size $p$ can be partitioned into $d$ sub-objects of respective size $\gamma_k$, for $k \in \lrb{1, d}$. 
We call it the \emph{${\gamma}$-composition} of an object. We give two examples. 
Let ${Q} \in \O(p)$. The {${\gamma}$-composition} of ${Q}$ is the sequence ${Q}^{{\gamma}} \coloneqq \lrp{{Q}_1, \dots, {Q}_d}$ such that $~{{Q}_k \in \R^{p\times\gamma_k}, \forall k \in \lrb{1, d}}$ and $~{{Q} = \lrb{{Q}_1 | \dots | {Q}_d}}$.
Let ${L} \coloneqq \lrp{\ell_1, \dots, \ell_p}$ be a sequence of decreasing eigenvalues. The {${\gamma}$-composition} of ${L}$ is the sequence  ${{L}}^{{\gamma}} \coloneqq \lrp{{{L}}_1, \dots, {{L}}_d}$ such that $~{{{L}}_k \in \R^{\gamma_k}}, ~{\forall k \in \lrb{1, d}}$ and ${L} = \lrb{{{L}}_1 |  \dots | {{L}}_d}$. 
We call \emph{${\gamma}$-averaging} of ${L}$ the sequence $\overline{{{L}}^{{\gamma}}} \coloneqq \lrp{\overline{{{L}}_1},  \dots, \overline{{{L}}_d}}\in \R^d$ of average eigenvalues in ${{L}}^{{\gamma}}$.

\subsection{Maximum likelihood}
Similarly as for PPCA, the log-likelihood of the model can be easily computed
\begin{equation}\label{appeq:PSA_LL}
\ln {{\mathcal{L}}} \lrp{{\mu}, {\Sigma}} = -\frac n 2 \lrp{p \ln(2\pi) + \ln |{\Sigma}| + \tr{{\Sigma}^{-1} {{C}}}},
\end{equation}
with
${{C}} = \frac 1 n \sum_{i=1}^n ({x}_i - {\mu}) {({x}_i - {\mu})}\T$.
We now show that the maximum likelihood estimate for PSA consists in the eigenvalue decomposition of the sample covariance matrix followed by a block-averaging of adjacent eigenvalues such that the imposed type ${\gamma}$ is respected; in other words, a ${\gamma}$-averaging of the eigenvalues.
Before that, let us naturally extend the notion of \emph{type} to symmetric matrices, as the sequence of multiplicities of its ordered-descending eigenvalues.
\begin{theorem}
\label{appthm:PSA}
Let $\lrp{x_i}_{i=1}^n$ be a p-dimensional dataset, $\overline {{x}} \coloneqq \frac 1 n \sum_{i=1}^n x_i$ its mean and $~{{S} \coloneqq \sum_{j=1}^p \ell_j {v}_j {{v}_j}\T}$ its sample covariance matrix, with $~{\ell_1 \geq \dots \geq \ell_p \geq 0}$ its eigenvalues and $\lrb{{v}_1 | \dots | {v}_p} \coloneqq {V} \in \O(p)$ its eigenvectors. 
The maximum likelihood parameters of PSA are
\begin{equation}
    \hat {{\mu}} = \overline {{x}}, \hspace{15mm}
    \hat {{Q}} = {{V}}, \hspace{15mm}
    \lrp{\hat{\lambda}_1, \dots, \hat{\lambda}_d} = \overline{{\lrp{\ell_1, \dots, \ell_p}}^{{\gamma}}}.
\end{equation}
The parameters $\hat {{\mu}}$ and $\hat{\lambda}_1, \dots, \hat{\lambda}_d$ are unique. $\hat {{Q}}$ is not unique but the flag of linear subspaces generated by its ${\gamma}$-composition is ``practically'' unique. More precisely, the flag is unique if and only if the type of $S$ is a refinement of $\gamma$, which is almost sure when $S$ is full-rank---when $S$ is rank-deficient, this is almost sure as long as all the null eigenvalues are gathered in the same subspace.
\end{theorem}

\begin{proof}
Original results about the maximum likelihood estimation of covariance eigenvalues and eigenvectors from multivariate Gaussian distributions with repeated covariance eigenvalues date back from the celebrated paper of~\cite{anderson_asymptotic_1963}.
We provide an independent proof for completeness with a particular emphasis on geometry, flags of linear subspaces, and uniqueness.
We successively find the optimal $\hat {{\mu}} \in \R^p$, $\hat {{Q}} \in \O(p)$ and $\hat{\lambda}_k \in \R$.

The log-likelihood expresses as a function of ${{\mu}} \in  \R^p$ in the following way
\begin{equation}\ln \mathcal{L}({{\mu}}) = -\frac n 2 \tr{{{\Sigma}}^{-1} {{C}}} + \operatorname{constant},
\end{equation}
with ${{C}} = \frac 1 n \sum_{i=1}^n ({{x}}_i - {{\mu}}) ({{x}}_i - {{\mu}})\T$.
The optimal shift $\hat{{{\mu}}}$ is thus
\begin{equation}
\hat{{{\mu}}} = \argmin_{\substack{{{\mu}} \in \R^p}}
\sum_{i=1}^n ({{x}}_i - {{\mu}})\T {{\Sigma}}^{-1} ({{x}}_i - {{\mu}}) \coloneqq f({{\mu}}).
\end{equation}
The gradient of ${{x}} \mapsto ({{x}} - {{\mu}})\T {{\Sigma}}^{-1} ({{x}} - {{\mu}})$ is ${{x}} \mapsto 2 {{\Sigma}}^{-1} ({{x}} - {{\mu}})$.
Hence, setting the gradient of $f$ to $0$ at $\hat {{\mu}}$, one gets
$~{
\sum_{i} 2 {{\Sigma}}^{-1} ({{x}}_i - \hat {{\mu}}) = 0
}$,
whose solution is $\hat {{\mu}} = \bar {{x}}$.
Hence $\hat {{C}}$ is actually the sample covariance matrix of the dataset, which will be denoted ${{S}}$ (as in the theorem statement) from now on.

The log-likelihood expresses as a function of ${{Q}}$ in the following way
\begin{equation}\ln \mathcal{L}({{Q}}) = -\frac n 2 \lrp{\ln |{{\Sigma}}| + \tr{{{\Sigma}}^{-1} {{S}}}} + \operatorname{constant},
\end{equation}
with ${{\Sigma}} = {{Q}} {{\Lambda}} {{Q}}\T$.
Hence $|{{\Sigma}}|$ is independent of ${{Q}}$ and the optimal orthogonal transformation $\hat{{{Q}}}$ is
\begin{equation}
\hat{{{Q}}} = \argmin_{\substack{{{Q}} \in \O(p)}} \tr{{{\Sigma}}^{-1} {{S}}} = \tr{{{Q}} {{\Lambda}}^{-1} {{Q}}\T {{S}}}  \coloneqq  g({{Q}}).
\end{equation}
As $g$ is a smooth function on $\O(p)$ which is a compact manifold, $\hat{{{Q}}}$ exists and 
\begin{equation}
    dg_{\hat{{{Q}}}}\colon \mathcal{T}_{\hat {{Q}}}(\O(p)) \ni {{\delta}} \mapsto \tr{\lrp{{{\delta}} {{\Lambda}}^{-1} {\hat {{Q}}}\T + {\hat {{Q}}} {{\Lambda}}^{-1} {{\delta}}\T} {{S}}} \in \R    
\end{equation}
vanishes.
It is known that $~{\mathcal{T}_{\hat {{Q}}}(\O(p)) = \Skew_p {\hat {{Q}}}}$, therefore one has for all ${{A}} \in \Skew_p$
\begin{equation}
dg_{\hat{{{Q}}}}({{A}} {\hat {{Q}}}) = 
\tr{\lrp{({{A}} {\hat {{Q}}}) {{\Lambda}}^{-1} {\hat {{Q}}}\T + {\hat {{Q}}} {{\Lambda}}^{-1} ({{A}} {\hat {{Q}}})\T} {{S}}} = \tr{{{A}} ({{\Sigma}}^{-1} {{S}} - {{S}} {{\Sigma}}^{-1})} = 0.
\end{equation}
Therefore ${{\Sigma}}^{-1} {{S}} - {{S}} {{\Sigma}}^{-1} = 0$. Hence, ${{S}}$ and  ${{\Sigma}}^{-1}$ are two symmetric matrices that commute, so they must be simultaneously diagonalizable in an orthonormal basis.
Since the trace is basis-invariant, $g$ simply rewrites as a function of the eigenvalues
\begin{equation}\label{appeq:g_psi}
    g({{Q}}) = \sum_{k=1}^d \lambda_k^{-1} \lrp{\sum_{j \in \phi_{{\gamma}}^{-1} (\lrs{k})} \ell_{\psi(j)}},
\end{equation}
where $\psi \in S_p$ is a permutation and $\phi_{{\gamma}}^{-1} (\lrs{k})$ is the set of indexes in the $k$-th part of the composition ${\gamma}$.
We now need to find the permutation $\hat \psi \in S_p$ that minimizes $g$.
First, since $~{\lambda_1 > \dots > \lambda_d > 0}$ by assumption, then $~{(\lambda_1^{-1}, \dots, \lambda_d^{-1}})$ is an increasing sequence.
Therefore, $(\ell_{\hat{\psi}(\phi_{{\gamma}}^{-1} \lrs{1})}, \dots, \ell_{\hat{\psi}(\phi_{{\gamma}}^{-1} \lrs{d})})$ must be a non-increasing sequence, in that for $k_1 < k_2$, the eigenvalues in the $k_1$-th part of ${\gamma}$ must
be greater than or equal to the eigenvalues in the $k_2$-th part.
Indeed, for $\lambda < \lambda'$, if $\ell < \ell'$, then $~{\lambda \ell' + \lambda' \ell < \lambda \ell + \lambda' \ell'}$.
Second, for such a $\hat{\psi}$ sorting the eigenvalues in non-increasing order in between parts, we can easily check that the inequality between eigenvalues of distinct parts is strict if and only if the type of ${{\Sigma}}$ is a refinement of ${{\gamma}}$. 
If so, the minimizing $\hat{\psi}$ is unique up to permutations within each part of ${\gamma}$. 
Therefore, it is not $\hat {{Q}}$ itself but the sequence of eigenspaces of $\hat {{Q}}$ generated by its ${\gamma}$-composition that is unique, and we have ${(\operatorname{Im}(\hat{{{Q}}}_1), \dots, \operatorname{Im}(\hat{{{Q}}}_d)) = (\operatorname{Im}(V_1), \dots, \operatorname{Im}(V_d))}$. Hence, the accurate space to describe the parameter $\hat {{Q}}$ is actually the space of flags of type $\gamma$.

An important remark is that the uniqueness condition will almost surely be met when $S$ is full-rank.
Indeed, the set of $p \times p$ symmetric matrices with repeated eigenvalues has null Lebesgue measure (it is a consequence of Sard's theorem applied to the discriminant polynomial function (as defined in~\cite{breiding_geometry_2018}).
Therefore, since sample covariance matrices are measurable functions with absolutely continuous (Gaussian) densities with respect to Lebesgue measure, a randomly drawn matrix ${{S}}$ almost surely has distinct eigenvalues. 
Consequently, its type is $\lrp{1, \dots, 1}$, which is a refinement of any possible type in $\mathcal{C}(p)$.
Note that the full-rank assumption avoids having multiple null eigenvalues with nonzero measure.

The log-likelihood expresses as a function of ${{\Lambda}}$ in the following way
\begin{equation}\ln \mathcal{L}({{\Lambda}}) = -\frac n 2 \lrp{\ln |{{\Sigma}}| + \tr{{{\Sigma}}^{-1} {{S}}}} + \operatorname{constant},
\end{equation}
with ${{\Sigma}} = \hat {{Q}} {{\Lambda}} {\hat {{Q}}}\T$. 
First, one has $\ln|{{\Sigma}}| = \sum_{k=1}^d \gamma_k \ln \lambda_k$. 
Second, according to the previous results, one has $\tr{{{\Sigma}}^{-1} {{S}}} = \sum_{k=1}^d \lambda_k^{-1} \lrp{\sum_{j \in \phi_{{\gamma}}^{-1} \lrs{k}} \ell_j}$.
The optimal eigenvalues $\lrp{\hat{\lambda}_1, \dots, \hat{\lambda}_d}$ are thus
\begin{equation}
\lrp{\hat{\lambda}_1, \dots, \hat{\lambda}_d} = \argmin_{\substack{\lambda_1, \dots, \lambda_d \in \R}} \sum_{k=1}^d \gamma_k \ln \lambda_k + \lambda_k^{-1} \lrp{\sum_{j \in \phi_{{\gamma}}^{-1} \lrs{k}} \ell_j}  \coloneqq  h(\lambda_1, \dots, \lambda_d).
\end{equation}
As 
$\frac {\partial h} {\partial \lambda_k} = \frac{\gamma_k}{\lambda_k} - \lambda_k^{-2} \lrp{\sum_{j \in \phi_{{\gamma}}^{-1} \lrs{k}} \ell_j}$, we get that $\hat {\lambda}_k = \frac 1 {\gamma_k} \lrp{\sum_{j \in \phi_{{\gamma}}^{-1} \lrs{k}} \ell_j} = \overline{L_k}$.
\end{proof}
One can then easily express the maximum log-likelihood of PSA
\begin{equation}\label{appeq:PSA_ML}
    \ln \hat{\mathcal{L}}(\gamma) = -\frac n 2 \lrp{p \ln(2\pi) + \sum_{k=1}^d \gamma_k \ln{\overline{L_k}} + p}.
\end{equation}

\subsection{Geometric interpretation with flag manifolds}
As discussed in the previous subsections, the appropriate parameter space for ${Q}$ in PSA is the space of flags of type ${\gamma}$, noted $\operatorname{Flag}({\gamma})$.
The geometry of such a space is well known~\cite{ye_optimization_2022}.
In a few words, each subspace $\mathcal{V}_k$ of dimension $\gamma_k$ can be endowed with an orthonormal basis $Q_k \coloneqq [q_k^1|\dots|q_k^{\gamma_k}]\in \R^{p\times\gamma_k}$. This basis is invariant to rotations within the subspace---i.e. for  $R_k\in\O(\gamma_k)$, $Q'_k \coloneqq Q_k R_k$ is still an orthonormal basis of $\mathcal{V}_k$. Concatenating such orthonormal frames for all the mutually-orthogonal subspaces of a flag creates an orthogonal matrix $Q \coloneqq [Q_1|\dots|Q_d]\in\O(p)$.
Eventually, $\operatorname{Flag}({\gamma})$ is a smooth quotient manifold, consisting in equivalence classes of orthogonal matrices:
\begin{equation}\label{eq:quotient}
    \operatorname{Flag}({\gamma}) \cong \O(p) / \lrp{\O(\gamma_1) \times \dots \times \O(\gamma_d)}.
\end{equation}
This result enables the accurate computation of the number of parameters in PSA.
Before that, let us note that the other parameters are ${\mu} \in \R^p$ and $~{{\Lambda} \in \operatorname{D}({\gamma}) \coloneqq \lrs{\diag{\lambda_1 {I}_{\gamma_1}, \dots, \lambda_d {I}_{\gamma_d}} \in \R^{p \times p} \colon \lambda_1 > \dots > \lambda_d > 0}}$, which can be seen as a convex cone of $\R^d$.
\begin{proposition}
\label{appprop:PSA_nparam}
The number of free parameters in PSA is
\begin{equation}\label{appeq:PSA_kappa}
    \kappa(\gamma) \coloneqq p + d + \frac{p(p-1)}{2} - \sum_{k=1}^{d} \frac {\gamma_k (\gamma_k - 1)} {2}.
\end{equation}
\end{proposition}
This geometric interpretation sheds light on PPCA, which---we remind---is a special case of PSA with ${\gamma} = \lrp{1, \dots, 1, p-q}$. First, as flags of type $(1, \dots, 1, p-q)$ belong to Stiefel manifolds (up to changes of signs), we can naturally parameterize PPCA models with those spaces, which is already commonly done in the literature~\cite{minka_automatic_2000}. Second, we can now see PPCA as removing $~{(p - q - 1) + \frac{(p-q)(p-q-1)}{2}}$ parameters with respect to the full covariance model by imposing an isotropy constraint on the noise space. PSA then goes beyond the noise space and results in even more parsimonious models.

We can extend this analysis to the IPPCA model, which---we remind---is a special case of PSA with ${\gamma} = \lrp{q, p-q}$. Hence we can parameterize it with flags of type $\lrp{q, p-q}$, which belong to Grassmannians.
With that in mind, we notice that our formula~\eqref{appeq:PSA_kappa} differs from the one given in~\cite{bouveyron_intrinsic_2011}. We think that this paper overestimates the number of free parameters by implicitly assuming eigenvectors living on Stiefel manifolds like in PPCA, whereas the isotropy in the signal space yields an additional rotational invariance which makes them actually live on Grassmannians. Therefore IPPCA is even more parsimonious than originally considered.

\section{Model selection}\label{appsec:MS}
As discussed previously, sample covariance matrices almost surely have distinct eigenvalues. This makes the full covariance model the most likely to have generated some observed data.
However, it does not mean that the true parameters---that are the eigenvectors and the eigenvalues---can be individually precisely inferred, especially in the small-data regime.
Hence, one can wonder if a covariance model with repeated eigenvalues and multidimensional eigenspaces would not be more robust.
The results of the previous section enable us to provide a possible answer, through PSA model selection. 
First, we study the inference of two adjacent eigenvalues and their associated eigenvectors. We show that when the relative eigengap is small and the number of samples is limited, one should prefer a PSA model with repeated eigenvalues---i.e. block-average the eigenvalues and gather the associated eigenvectors in a multidimensional eigenspace.
Second, to extend this result to more than two eigenvalues, we develop a general model selection framework based on the stratified structure of PSA models.

\subsection{Bayesian information criterion}
The Bayesian information criterion (BIC) is defined as 
\begin{equation}\label{appeq:BIC}
    \operatorname{BIC}(\gamma) = \kappa(\gamma) \ln n - 2 \ln \hat{\mathcal{L}}(\gamma),
\end{equation}
where $\kappa$ is the number of free parameters~\eqref{appeq:PSA_kappa} and $\ln \hat{\mathcal{L}}$ is the maximum log-likelihood~\eqref{appeq:PSA_ML}.
It is a widely-used model selection criterion, making a tradeoff between model complexity $\kappa$ and goodness-of-fit $\hat{\mathcal{L}}$. The formula results from an asymptotic approximation of the model evidence.
In this section, we use the BIC for PSA model selection. The model with lowest BIC is considered as the best model. In the two-eigenvalue case, we get an explicit criterion based on eigenvalue gaps to decide if we must assume that they are equal, and in the more general case, we propose efficient model comparison strategies. We also investigate other model selection criteria than the BIC for completeness in this section, and get similar conclusions.

\subsection{The two-eigenvalue case}
In order to better understand the dynamics of PSA model selection, we lead the experiment of quantifying the BIC variation induced by the equalization of two adjacent eigenvalues. More precisely and without loss of generality, we compare the BIC of a \emph{full covariance model} ${\gamma} = \lrp{1, \dots, 1}$ to the one of an \emph{equalized covariance model} ${\gamma}' = \lrp{1 \dots 1, 2, 1 \dots 1}$, where the eigenvalue $\lambda_j$ has multiplicity $2$.

\begin{theorem}
\label{appthm:releigengap}
Let $\lrp{x_i}_{i=1}^n$ be a $p$-dimensional dataset with $n$ samples, $\ell_j \geq \ell_{j+1}$ two adjacent sample eigenvalues and $\delta_j = \frac{\ell_{j} - \ell_{j+1}}{\ell_j}$ be their \emph{relative eigengap}. 
If
\begin{equation}\label{appeq:releigengap_BIC}
    \delta_j < 2\lrp{1 - n^{\frac 2 n} + n^{\frac 1 n}\sqrt{n^{\frac 2 n} - 1}},
\end{equation}
then the equalized covariance model has a lower BIC than the full one.
\end{theorem}
\begin{proof}
Since $n$ and $p$ are constant within model selection, the BIC can be rewritten (up to constant terms and factors) as
\begin{equation}\label{appeq:BIC_simpl}
\operatorname{BIC} (\gamma) \coloneqq \lrp{d - \sum_{k=1}^d \frac{\gamma_k (\gamma_k - 1)} {2}} \frac{\ln{n}}{n} + \sum_{k=1}^d \gamma_k \ln{\overline{L_k}}.
\end{equation}
We compare the BIC of the full covariance model ${\gamma} = \lrp{1, \dots, 1}$ to the one of the equalized covariance model ${\gamma}' = \lrp{1, \dots, 1, 2, 1, \dots 1}$ where the $j$-th eigenvalue has been equalized with the $j+1$-th. This boils down to studying the sign of the function $\Delta \BIC = \BIC({{\gamma}}) - \BIC({{\gamma}}')$. One gets  
\begin{align}
\Delta \BIC &= p \frac{\ln n}{n} + \sum_{k=1}^p \ln \ell_k - \lrp{p - 2} \frac{\ln n}{n} - \sum_{k \notin \lrs{j, j+1}} \ln \ell_k - 2 \ln\lrp{\frac{\ell_j + \ell_{j+1}} 2},\\
&= 2\frac{\ln n}{n} + \ln \ell_j + \ln \ell_{j+1} - 2 \ln\lrp{\frac{\ell_j + \ell_{j+1}} 2},\\
&= 2\frac{\ln n}{n} + \ln \ell_j + \ln \lrp{\ell_j \lrp{1 - \delta_j}} - 2 \ln\lrp{\frac{\ell_j \lrp{2 - \delta_j}} 2},\\
&= 2\frac{\ln n}{n} + \ln{\lrp{1 - \delta_j}} - 2 \ln\lrp{1 - \frac{\delta_j} 2},\\
&= 2\frac{\ln n}{n} - \ln\lrp{\frac{\lrp{1 - \frac{\delta_j} 2}^2}{1 - \delta_j}}.
\end{align}
Hence, one has
\begin{equation}
\Delta \BIC = 0 \iff \exp\lrp{2\frac{\ln n}{n}} = \frac{\lrp{1 - \frac{\delta_j} 2}^2}{1 - \delta_j} \iff \frac {\delta_j^2} 4 - \lrp{1 - \exp\lrp{2 \frac{\ln n} n}} \delta_j + 1 - \exp\lrp{2 \frac{\ln n} n} = 0.    
\end{equation}
It is a polynomial equation whose positive solution is unique when $n \geq 1$ and is
\begin{equation}
\delta(n) =  2 - 2 \exp\lrp{2 \frac{\ln n} n} + 2\sqrt{\exp\lrp{4 \frac{\ln n} n} - \exp\lrp{2 \frac{\ln n} n}}.
\end{equation}
\end{proof}

\subsection{Comparison with North's rule-of-thumb}~\label{appsubsec:north}
A rule-of-thumb for determining which sample eigenvalue pairs might lead to large PC sampling error is proposed in~\cite{north_sampling_1982}.
The authors show that the asymptotic sampling error of a population eigenvalue $\lambda$ is $\Delta\lambda \coloneqq \lambda(\frac2n)^{\frac12}$ in the Gaussian setting. North's rule-of-thumb (NRT) states that when one population eigenvalue's sampling error is comparable to or larger than its distance to an adjacent eigenvalue, then the PC's sampling error is comparable to the associated adjacent PC.
Note that this is not an explicit rule (compared to our relative eigengap threshold~\eqref{appeq:releigengap_BIC}) since one has to choose the level of uncertainty, and---most of all---it is based on the \textit{true} eigenvalues (on which the confidence intervals are based) which are unknown.
However, this rule has been applied in many contexts and it is commonly implemented in the following way~\cite{sinkr}. 
For each sample eigenvalue pair $\ell_{j} \geq \ell_{j+1}$, compute the 1 sigma error intervals $I_{j} = [\ell_{j} - \ell_{j}\sqrt{\frac2n}, \ell_{j} + \ell_{j}\sqrt{\frac2n}]$ and $I_{j+1} = [\ell_{j+1} - \ell_{j+1}\sqrt{\frac2n}, \ell_{j+1} + \ell_{j+1}\sqrt{\frac2n}]$. If $I_j\cap I_{j+1} \neq \emptyset$, then the associated principal components suffer from large sampling errors and might be random mixtures of the true eigenvectors. We reformulate it as a relative eigengap threshold.
\begin{proposition}\label{appprop:releigengap_North}
North's rule-of-thumb (as implemented in practice) boils down to the relative eigengap threshold
\begin{equation}\label{appeq:releigengap_North}
    \delta_j \leq \frac{2\sqrt{\frac2n}}{1+\sqrt{\frac2n}}.
\end{equation}
\end{proposition}
\begin{proof}
The sampling error interval overlap condition writes as
\begin{align}
    \ell_{j} - \sqrt{\frac2n} \ell_{j} \leq \ell_{j+1} + \sqrt{\frac2n} \ell_{j+1} & \iff \frac{\ell_j - \ell_{j+1}}{\ell_j} \leq \sqrt{\frac2n} \lrp{1 + \frac{\ell_{j+1}}{\ell_j}},\\
    & \iff \frac{\ell_j - \ell_{j+1}}{\ell_j} \leq \sqrt{\frac2n} \lrp{2 - \frac{\ell_j - \ell_{j+1}}{\ell_j}},\\
    & \iff \frac{\ell_j - \ell_{j+1}}{\ell_j} \leq \frac{2 \sqrt{\frac2n}}{1 + \sqrt{\frac2n}}.
\end{align}
\end{proof}
This threshold is reported in Fig~\ref{fig:releigengap_curves}, under the name NRT-1 (for 1 sigma sampling errors). We also report North's rule-of-thumb for 2 sigma sampling errors (NRT-2), yielding a relative eigengap threshold of $\frac{4 \sqrt{\frac2n}}{1 + 2\sqrt{\frac2n}}$.
We see that the relative eigengap NRT-1 is much smaller than ours (e.g. $8.6\%$ instead of $21\%$ for $1000$ samples). Therefore, although warning scientists about close sample eigenvalues in principal component analysis, North's rule-of-thumb largely overlooks the curse of isotropy compared to our method.
To see the practical effect of this lower threshold, we test this condition on the same real datasets as in Fig~\ref{fig:releigengap_UCI}. The results are in Fig~\ref{appfig:releigengaps_real}.
\begin{figure}
    \centering
    \includegraphics[width=\linewidth]{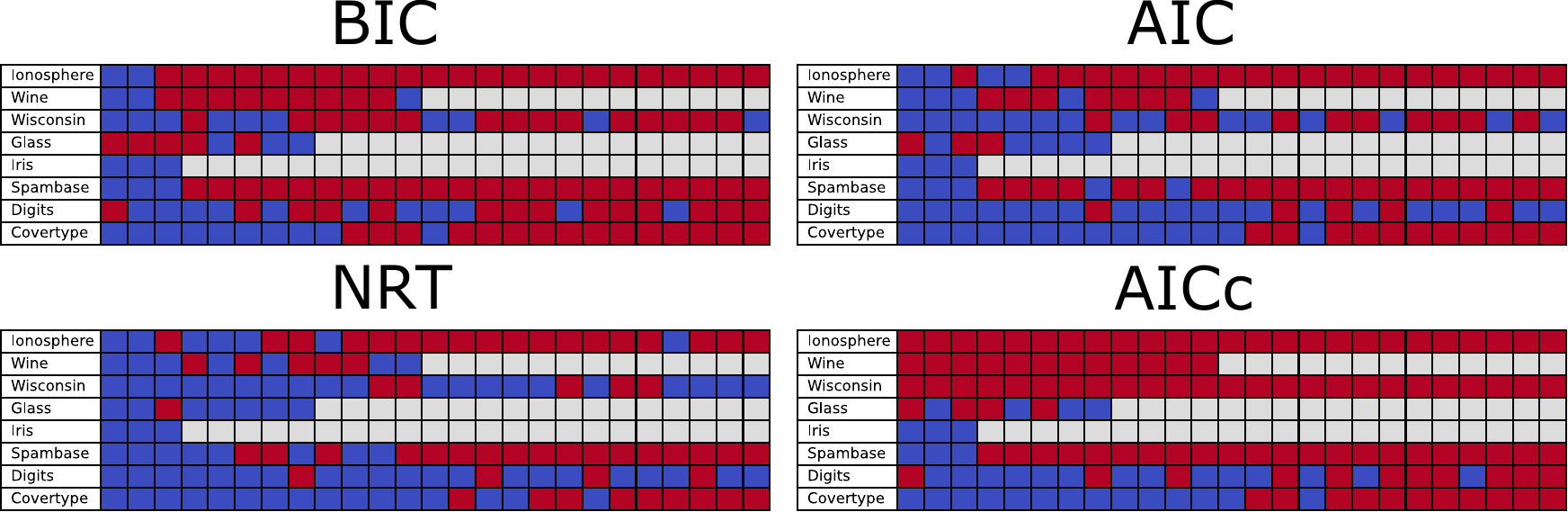}
    \caption{Practical effect of the different relative eigengap conditions using the BIC (Theorem~\ref{appthm:releigengap}), NRT (Proposition~\ref{appprop:releigengap_North}), AIC (Proposition~\ref{appprop:releigengap_AIC}) and AICc (Proposition~\ref{appprop:releigengap_AICc}) for several classical datasets from the UCI Machine Learning repository. A red case in column $j$ indicates that the $(j, j+1)$ eigenvalue pair is below the relative eigengap threshold and should be equalized. Blue indicates above, and gray that the pair does not exist. We only plot the $25$ leading eigenvalue pairs.
    We can see that all the methods suggest to consider principal subspaces of dimension greater than $1$, including North's rule-of-thumb which has the lowest relative eigengap thresholds (cf. Fig~\ref{fig:releigengap_curves}).
    Moreover, the number of eigenvalues to equalize seems to increase from NRT to AIC to BIC, but the low-sample correction of AIC seems to equalize even more eigenvalues than the BIC. This stresses that accounting for low-sample sizes is an important issue in the curse of isotropy. 
    }
	\label{appfig:releigengaps_real}
\end{figure}
We can see that the curse of isotropy remains a nonnegligible phenomenon with North's rule, even though it is less marked than with the BIC.
We think that North's rule (as implemented in practice) underestimates the phenomenon, notably because it uses 1 sigma uncertainties and since it is based on sample eigenvalues instead of true eigenvalues in the implementations. We recall that 1 sigma uncertainties (NRT-1) correspond to $68\%$ error bars while 2 sigma uncertainties (NRT-2) correspond to $95\%$ error bars and yield a relative eigengap threshold of $16\%$, which is much closer to our results with the BIC.
An interesting perspective would be to consider our guideline instead of the less-impactful North's rule in seminal climate science papers which made some conclusions out of possibly degenerate principal components.

\subsection{Comparison with other model selection criteria}
Although being widely used in model selection, the BIC is well-known for its heavy complexity penalization, tending to select over-parsimonious models~\cite{burnham_model_2004}. Another widely-used criterion is the Akaike information criterion~\cite{akaike_new_1974}. It is defined as
\begin{equation}\label{appeq:AIC}
    \operatorname{AIC}(\gamma) = 2 \kappa(\gamma) - 2 \ln \hat{\mathcal{L}}(\gamma)
\end{equation}
where $\kappa$ is the number of free parameters~\eqref{appeq:PSA_kappa} and $\ln \hat{\mathcal{L}}$ is the maximum log-likelihood~\eqref{appeq:PSA_ML}.
Comparing an equalized covariance model to one with distinct eigenvalues like in Theorem~\ref{appthm:releigengap} but this time using the AIC yields another relative eigengap condition.
\begin{proposition}\label{appprop:releigengap_AIC}
Let $\lrp{x_i}_{i=1}^n$ be a $p$-dimensional dataset with $n$ samples, $\ell_j \geq \ell_{j+1}$ two adjacent sample eigenvalues and $\delta_j = \frac{\ell_{j} - \ell_{j+1}}{\ell_j}$ their \emph{relative eigengap}. 
If
\begin{equation}\label{appeq:releigengap_AIC}
    \delta_j < 2\lrp{1 - e^{\frac4n} + e^{\frac2n}\sqrt{e^{\frac4n} - 1}}
\end{equation} 
then the equalized covariance model has a lower AIC than the full one.
\end{proposition}
\begin{proof}
The proof is essentially the same as the one of Theorem~\ref{appthm:releigengap}.
Since $n$ and $p$ are constant within model selection, the AIC can be rewritten (up to constant terms and factors) as
\begin{equation}\label{appeq:AIC_simpl}
\operatorname{AIC} (\gamma) \coloneqq \lrp{d - \sum_{k=1}^d \frac{\gamma_k (\gamma_k - 1)} {2}} \frac{2}{n} + \sum_{k=1}^d \gamma_k \ln{\overline{L_k}}
\end{equation}
Replacing $\frac{\ln n}{n}$ with $\frac{2}{n}$ in the proof of Theorem~\ref{appthm:releigengap}, we finally get the result that 
\begin{equation}
    \delta(n) =  2 - 2 \exp\lrp{\frac{4} n} + 2\sqrt{\exp\lrp{\frac{8} n} - \exp\lrp{\frac{4} n}}.
\end{equation}
\end{proof}
This threshold is reported in Fig~\ref{fig:releigengap_curves}. We see that this relative eigengap is smaller than ours~\eqref{appeq:releigengap_BIC} (e.g. $12\%$ instead of $21\%$ for $1000$ samples), but {larger} than North's rule~\eqref{appeq:releigengap_North}. This result is interesting since AIC is known for tending to select overparameterized models, especially for small sample sizes~\cite{burnham_model_2004} (cf. next paragraph). 
Despite this, the relative eigengap condition with AIC is more impactful than North's rule.
To see the practical effect of the AIC threshold of~\eqref{appeq:releigengap_AIC}, we also report the relative eigengap condition on real datasets in Fig~\ref{appfig:releigengaps_real}.
We see that many eigenvalue pairs should be assumed equal---slightly less than with BIC. Therefore, even with another model selection criterion, the curse of isotropy is still a nonnegligible phenomenon in real datasets, and the principal subspace analysis methodology enables to leverage it to improve interpretability.

Additionally, we provide a relative eigengap condition for the AICc~\cite{hurvich_regression_1989}, which is a small-sample correction to the AIC. In practice, the AICc is advised over the AIC for $n/\kappa < 40$~\cite{burnham_model_2004}.
The AICc is defined as
\begin{equation}\label{appeq:AICc}
    \operatorname{AICc}(\gamma) = 2 \kappa(\gamma) {\frac{n}{n-\kappa(\gamma)-1}} - 2 \ln \hat{\mathcal{L}}(\gamma)
\end{equation}
where $\kappa$ is the number of free parameters~\eqref{appeq:PSA_kappa} and $\ln \hat{\mathcal{L}}$ is the maximum log-likelihood~\eqref{appeq:PSA_ML}.
One can see that this corrected criterion converges asymptotically to the AIC. 
Comparing an equalized covariance model to one with distinct eigenvalues like in Theorem~\ref{appthm:releigengap} but this time with the AICc yields the following relative eigengap condition.
\begin{proposition}\label{appprop:releigengap_AICc}
Let $\lrp{x_i}_{i=1}^n$ be a $p$-dimensional dataset with $n > \frac{p(p+3)}{2} + 1$ samples, $\ell_j \geq \ell_{j+1}$ two adjacent sample eigenvalues, $\delta_j = \frac{\ell_{j} - \ell_{j+1}}{\ell_j}$ their \emph{relative eigengap} and $\varphi = \frac{4n-4}{\lrp{n - \frac{p(p+3)}{2}}^2 - 1}$.
If
\begin{equation}\label{appeq:releigengap_AICc}
    \delta_j < 2\lrp{1 - e^{\varphi} + e^{\frac{\varphi}{2}}\sqrt{e^{\varphi} - 1}}
\end{equation} 
then the equalized covariance model has a lower AICc than the full one.
\end{proposition}
\begin{proof}
The proof is essentially the same as in Thm~\ref{appthm:releigengap} and Proposition~\ref{appprop:releigengap_AIC}.
Since $n$ and $p$ are constant within model selection, the AICc can be rewritten (up to constant terms and factors) as
\begin{equation}\label{appeq:AICc_simpl}
\operatorname{AICc} (\gamma) \coloneqq \frac{2\kappa(\gamma)}{n - \kappa(\gamma) - 1} + \sum_{k=1}^d \gamma_k \ln{\overline{L_k}}
\end{equation}
We compare the AICc of the full covariance model ${\gamma} = \lrp{1, \dots, 1}$ to the one of the equalized covariance model ${\gamma}' = \lrp{1, \dots, 1, 2, 1, \dots 1}$ where the $j$-th eigenvalue has been equalized with the $j+1$-th. This boils down to studying the sign of the function $\Delta \operatorname{AICc} = \operatorname{AICc}({{\gamma}}) - \operatorname{AICc}({{\gamma}}')$. One gets
\begin{align}
\Delta \operatorname{AICc} &= \frac{p(p+3)}{n - \frac{p(p+3)}{2} - 1} - \frac{p(p+3) - 4}{n - \lrp{\frac{p(p+3)}{2} - 2} - 1} + \ln \ell_j + \ln \ell_{j+1} - 2 \ln\lrp{\frac{\ell_j + \ell_{j+1}} 2}\\
&= \frac{4n-4}{\lrp{n-\frac{p(p+3)}{2}}^2 - 1} + \ln \ell_j + \ln \ell_{j+1} - 2 \ln\lrp{\frac{\ell_j + \ell_{j+1}} 2}
\end{align}
Replacing $2\frac{\ln n}{n}$ with $\varphi = \frac{4n-4}{\lrp{n-\frac{p(p+3)}{2}}^2 - 1}$ in the proof of Theorem~\ref{appthm:releigengap}, we finally get the result that 
\begin{equation}
    \delta(n) =  2 - 2 \exp\lrp{\varphi} + 2\sqrt{\exp\lrp{2\varphi} - \exp\lrp{\varphi}}.    
\end{equation}
\end{proof}
Contrary to the other criteria (Thm~\ref{appthm:releigengap}, Proposition~\ref{appprop:releigengap_North} and Proposition~\ref{appprop:releigengap_AIC}), this threshold depends on the dimension $p$. Therefore, we plot it for several $p$ in Fig~\ref{fig:releigengap_curves}. 
We can see that this relative eigengap converges to the AIC for large $n$, but is {larger} than the one with the BIC~\eqref{appeq:releigengap_BIC} when the number of samples is close to the number of model parameters.
We also test this condition on the same real datasets as in Fig~\ref{fig:releigengap_UCI} and report the results in Fig~\ref{appfig:releigengaps_real}.
We see that many eigenvalue pairs are ill-defined, especially in high-dimensional datasets where those are even more numerous than with the BIC.

\subsection{Efficient model selection}
Given a dimension $p$, PPCA has $p$ models, ranging from the isotropic Gaussian ($q=0$) to the full covariance model ($q=p-1$). We can naturally equip the set of PPCA models with the  \emph{less-than-or-equal} relation $\leq$ on the latent variable dimension $q$, which makes it a totally ordered set. The complexity of the model then increases with $q$.

The characterization of the PSA family structure is a bit more technical, as it requires to study the hierarchy of types, involving the concept of integer composition. Fortunately, this analysis can be lifted to the stratification of symmetric matrices according to the multiplicities of the eigenvalues, which is already well-known~\cite{arnold_modes_1972,groisser_geometric_2017,breiding_geometry_2018}. Therefore, without proof, we can state the following result.

\begin{proposition}\label{appprop:pos}
The family of $p$-dimensional PSA models induces a stratification of the space of symmetric positive-definite (SPD) matrices $S_p^{++}$ according to the type ${\gamma}$.
The refinement relation $\preceq$ makes it a partially ordered set of cardinal $2^{p-1}$.
\end{proposition}
Hence the set of PSA models at a given data dimension can be represented using a Hasse diagram, as done in Fig~\ref{fig:hasse_complexity}.
We see that PSA contains PPCA, IPPCA, and many new models. 
PSA therefore has the advantage of possibly providing more adapted models than PPCA and IPPCA, but also the drawback of requiring more comparisons for model selection. 
In high dimension{s} this becomes quickly computationally heavy, therefore we need to define strategies for selecting only a few number of models to compare. The previously derived partial order $\preceq$ on the set of PSA models allows simple efficient strategies for model selection. In the following subsubsections, we detail those strategies and prove additional properties.

\subsubsection{Relative eigengap threshold clustering of eigenvalues}
The \textit{relative eigengap threshold strategy} consists in clustering the eigenvalues whose relative eigengap $\delta_j \coloneqq \frac{\ell_{j} - \ell_{j+1}}{\ell_j}$ is below a given threshold, e.g. the one of Theorem~\ref{appthm:releigengap}. This clustering uniquely determines a PSA type $\gamma$, from which we apply maximum likelihood estimation, i.e. we block-average the corresponding eigenvalue clusters.
This rule is extremely simple but it may select overly parsimonious models, since distant eigenvalues may end up in the same cluster by propagation. Therefore, we provide a more-advanced strategy in the following subsubsection.

\subsubsection{Hierarchical clustering of eigenvalues}
In this strategy, the subset of candidate models is generated by the \emph{hierarchical clustering} of the sample eigenvalues. The general principle of hierarchical clustering is to agglomerate one by one the eigenvalues into clusters, thanks to a so-called \emph{cluster-linkage criterion}, which is a measure of dissimilarity between clusters.
More precisely, here we choose a \textit{continuous} pairwise distance $\delta$ between adjacent eigenvalues (such as the relative eigengap defined in Theorem~\ref{appthm:releigengap}), and a linkage criterion $\Delta$ between eigenvalue clusters, making sense with respect to our model selection problem (such as the single-linkage criterion $\Delta(\Lambda_1, \Lambda_2) = \min_{\ell_1, \ell_2 \in \Lambda_1 \times \Lambda_2} \delta(\ell_1, \ell_2)$ or the centroid-linkage criterion $\Delta(\Lambda_1, \Lambda_2) = \delta(\overline{\Lambda_1}, \overline{\Lambda_2})$). 
The method is detailed in Algorithm~\ref{alg:hierarchical} and illustrated in Fig~\ref{appfig:hierarchical_clustering}.
\begin{figure}
    \centering
    \includegraphics[width=\linewidth]{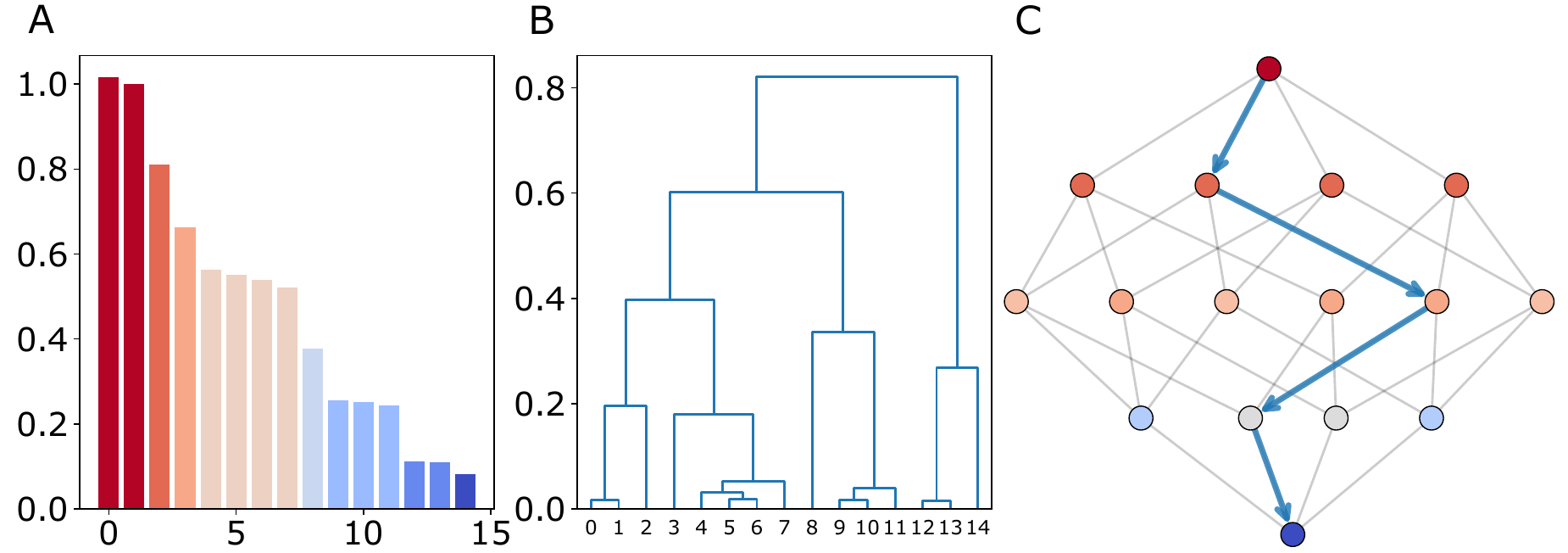}
    \caption{Hierarchical clustering of sample eigenvalues, using the relative eigengap distance for $\delta$ and the centroid-linkage criterion for $\Delta$.
    (A) Sample eigenvalues, whose colors correspond to a given step $t=8$ of the hierarchical clustering, with $~{{\gamma}^{t} = (2, 1, 1, 4, 1, 3, 2, 1)}$.
	(B) Hierarchical clustering dendrogram.
	(C) Conceptual illustration of the hierarchical clustering strategy. This heuristic generates a sequence of PSA models $({\gamma}^t)_{t=1}^p$ of  decreasing complexity, starting from the full covariance model and ending at the isotropic covariance model. This can be visualized as a trajectory in the Hasse diagram of PSA models (cf. Fig~\ref{fig:hasse_complexity}).
	}
	\label{appfig:hierarchical_clustering}
\end{figure}
The hierarchical clustering strategy creates a \emph{trajectory} $({\gamma}^t)_{t=1}^p$ in the Hasse diagram of PSA models (cf. Fig~\ref{fig:hasse_complexity}). The sequence starts from ${\gamma}^1 = \lrp{1, \dots, 1}$, the full covariance model, in which each eigenvalue is in its own cluster. Then, one by one, the eigenvalues that are the closest in terms of distance $\Delta$ are agglomerated, and the inter-cluster distances are updated. The algorithm ends when one reaches the isotropic covariance model, ${\gamma}^p = \lrp{p}$, in which all the eigenvalues are in the same cluster. This corresponds to an \textit{agglomerative} approach in the hierarchical clustering vocabulary, in opposition to a \textit{divisive} approach, that we could similarly develop for this strategy.

The hierarchical clustering strategy hence generates a subfamily of $p$ models that can be then compared within a classical model selection framework. In order to assess the quality of such a strategy, we show the following consistency result.

\begin{proposition}[Asymptotic consistency of the hierarchical clustering strategy]\label{appprop:hierarchical_heuristic}
The hierarchical clustering strategy generates a subfamily of PSA models that almost surely contains the true PSA model for $n$ large enough.
\end{proposition}
\begin{proof}
Let us assume that the true generative model is stratified with type ${{\gamma}} \in \mathcal{C}(p)$. 
We can then write the population covariance matrix as ${{\Sigma}} = \sum_{k=1}^{d} \lambda_k {{Q}}_k {{{Q}}_k}\T$ with $~{\lambda_1 > \dots > \lambda_{d} > 0}$ and ${{Q}} \coloneqq \lrb{{{Q}}_1|\dots|{{Q}}_{d}} \in \O(p)$. 
Let $n$ be the number of independent samples and ${{S}}_n \coloneqq \sum_{j=1}^{p} \ell_j({{S}}_n) {{v}}_j({{S}}_n) {{{v}}_j({{S}}_n)}\T$ with $\ell_1 \geq \dots \geq \ell_p$ and ${{V}} \coloneqq \lrb{{{v}}_1|\dots|{{v}}_p} \in \O(p)$. 
According to Tyler (1981), Lemma~2.1~(i), one then has almost surely, as $n$ goes to infinity, $\ell_j({{S}}_n) \to \lambda_{\phi_{{{\gamma}}}(j)}$, where $\phi_{{{\gamma}}}$ is the ${{\gamma}}$-composition function.
Hence for $n$ large enough, by continuity of the distance function $\Delta$, the gaps between eigenvalues in the same part of the ${{\gamma}}$-composition will be arbitrarily close to $0$, while the other will be arbitrarily close to the true values $\lrs{\Delta\lrp{\lambda_k, \lambda_{k+1}}, k \in \lrb{1, d-1}}$, which are all positive.
Hence the hierarchical clustering method will first agglomerate the eigenvalues that are in the same part of ${{\gamma}}$, and second the distinct blocks, by increasing order of pairwise distance. The last model of the first phase will be exactly the true model.
\end{proof}
Hence, the hierarchical clustering strategy generates a hierarchical subfamily of models of decreasing complexities, including the true PSA model for $n$ large enough. The true model can be then recovered using asymptotically consistent model selection criteria on the subfamily.
We now propose a second strategy that is not hierarchical but instead makes a prior assumption on the model complexity and then selects the one that has the maximum likelihood among all the candidates.

\subsubsection{Prior on the number of distinct eigenvalues}
In this strategy, we perform model selection at a given level of the Hasse diagram (cf. Fig~\ref{fig:hasse_complexity}). More precisely, we consider as candidates only the models that have a given type length $d$, like done in IPPCA with $d=2$.
The type-length prior strategy reduces the search space like the previous strategy, this time to $\binom{p-1}{d-1}$ models. In contrast to the hierarchical clustering strategy which creates a hierarchy of models with decreasing complexity, we here rather fix the complexity range of the candidate models, by working on one floor of the Hasse diagram, and then try to find the model of best fit.

Just like in the hierarchical clustering strategy, we could use the BIC to choose the best model among this reduced family.
For completeness, we provide an additional criterion that is nothing but the maximum likelihood itself. 
We indeed manage to extend to PSA the surprising result from~\cite{bouveyron_intrinsic_2011} stating that the maximum likelihood criterion alone asymptotically consistently finds the true intrinsic dimension within the IPPCA setting. 
Intuitively, this can be explained by the fact that we a priori fix the complexity of the candidate models and therefore we can focus on the other side of the weighing scale that is the goodness of fit.
As this criterion empirically yields competitive results with respect to other classical model selection criteria in the large sample, low signal-to-noise ratio regime, we expect it to be of interest in PSA as well.
\begin{proposition}[Asymptotic consistency of the maximum likelihood for fixed $d$]\label{appprop:fixed_length_heuristic}
If the true PSA model has $d$ distinct eigenvalues, then maximum likelihood model selection within the subfamily of PSA models of type-length $d$ almost surely recovers the true model for $n$ large enough.
\end{proposition}
\begin{proof}
Let us assume that the true generative model is stratified with type $~{{{\gamma}}^*\coloneqq\lrp{\gamma_1^*, \dots, \gamma_d^*}}$, of length $d$, and let $\lambda_1 > \dots > \lambda_{d} > 0$ be the eigenvalues of the associated population covariance matrix.
Then, similarly as in the previous proof, almost surely, asymptotically, the sample covariance matrix eigenvalues are the ones of the population covariance matrix.
Hence, for any PSA model of type ${{\gamma}} \coloneqq \lrp{\gamma_1, \dots, \gamma_d}$, the maximum likelihood writes
\begin{equation}
\ln{\hat{\mathcal{L}}}{(\gamma)} {=} -\frac n 2 \lrp{p \ln 2\pi + \sum_{k=1}^{d} \gamma_k \ln \lrp{\frac{1}{\gamma_k}\sum_{j \in \phi_{{{\gamma}}}^{-1} \lrs{k}} \lambda_{\phi_{{{\gamma}}^*}(j)}}}.
\end{equation}
As $n$ and $p$ are fixed when we compare the models, they do not intervene in the model selection. Hence, the search of the optimal model in terms of maximum likelihood boils down to the following problem 
\begin{equation}
    \argmin_{\substack{{{\gamma}} \in \mathcal{C}(p)\\ \#{{{\gamma}}}=d}}
    \sum_{k=1}^{d} \gamma_k \ln \lrp{\frac{1}{\gamma_k}\sum_{j \in \phi_{{{\gamma}}}^{-1} \lrs{k}} \lambda_{\phi_{{{\gamma}}^*}(j)}} \coloneqq f({{\gamma}}).
\end{equation}
One has $f({{\gamma}}) = \sum_{k=1}^{d} \gamma_k \ln (\frac{1}{\gamma_k}\sum_{k'=1}^{d}
c_{kk'} \lambda_{k'})$, where $c_{kk'}$ is the cardinal of the intersection of the $k$-th part of ${{\gamma}}$ with the $k'$-th part of ${{\gamma}}^*$.
Then, by definition, one has $\sum_{k'=1}^{d} c_{kk'} = \gamma_k$ and $\sum_{k=1}^{d} c_{kk'} = {{\gamma}}^*_{k'}$. Hence, using Jensen's inequality,
\begin{equation}
f({{\gamma}}) \geq \sum_{k=1}^{d} \gamma_k \lrp{\sum_{k'=1}^{d} \frac{c_{kk'}}{\gamma_k}\ln \lambda_{k'}} = \sum_{k,k'=1}^{d} c_{kk'} \ln \lambda_{k'} = \sum_{k'=1}^{d} {{\gamma}}^*_{k'} \ln \lambda_{k'} = f({{\gamma}}^*).
\end{equation}
To conclude, asymptotically, ${{\gamma}}^*$-PSA is the most likely model. Hence, the maximum likelihood criterion alone finds the true model among the family of PSA models with the same type length.
\end{proof}
Hence we derived three simple strategies for model selection, taking into account the structure of the PSA models family. 
\begin{remark}
Many variants can be adopted depending on the problem at hand. For instance if the noise is known, or assumed with some explained variance ratio rules, one can first search for the associated intrinsic dimension $q$ like in classical PCA, and then try to equalize some of the $q$ first eigenvalues by optimizing the model selection criterion over the subfamily of models whose $p - q$ last eigenvalues are all equal.
\end{remark}
\begin{remark}
In high dimensions, some eigenvalues might be very small or even null. The case of small positive eigenvalues may yield {large} relative eigengaps in the last eigenvalue pairs---therefore PSA model selection tends to separate those eigenvalues---whereas those are traditionally considered as noise. The case of null eigenvalues may even yield undefined PSA models. To circumvent those two issues, a classical trick is the one of \textit{covariance regularization}, consisting in adding a small constant to all the covariance eigenvalues. This somewhat boils down to adding an isotropic Gaussian noise to the data. This notably has the effect of diminishing the relative eigengaps, especially for the small positive or null eigenvalues. Another idea is to constrain the model types to have at least the last $p - q$ eigenvalues equal, where $q$ is chosen sufficiently small such that the first $q$ eigenvalues are sufficiently large.
\end{remark}

\section{{Statistical evaluation of the PSA methodology}}\label{appsec:evaluation}
A key result in the previous section is that we rarely have enough samples to confidently assert that two adjacent sample eigenvalues are distinct.
Consequently, PPCA models could be made more parsimonious by equalizing the adjacent sample eigenvalues with small gaps in the signal space as well. {In this section, we provide additional theoretical and experimental evidences for the interest of PSA over PPCA. We thank the anonymous reviewers for suggesting us to explore some of these insightful ideas.}

\subsection{{Model selection for increasing sample sizes}}
In order to better understand how {our relative eigengap} result{s} apply in practice, we make the following PSA model selection experiment.
We consider a given multivariate Gaussian population density, with covariance matrix eigenvalues $\lrp{10, 9, 7, 4, 0.5}$, and sample $n \in \lrb{20, 50000}$ data points from it. We fit all the PSA models to this data distribution and select the one with the lowest BIC. The experiment is repeated several times independently  for each $n$, and the results are reported in Fig~\ref{appfig:BIC_traj}, where we plot only a few models among the $16$ for readability.
\begin{figure}[t]
     \centering
     \includegraphics[width=\linewidth]{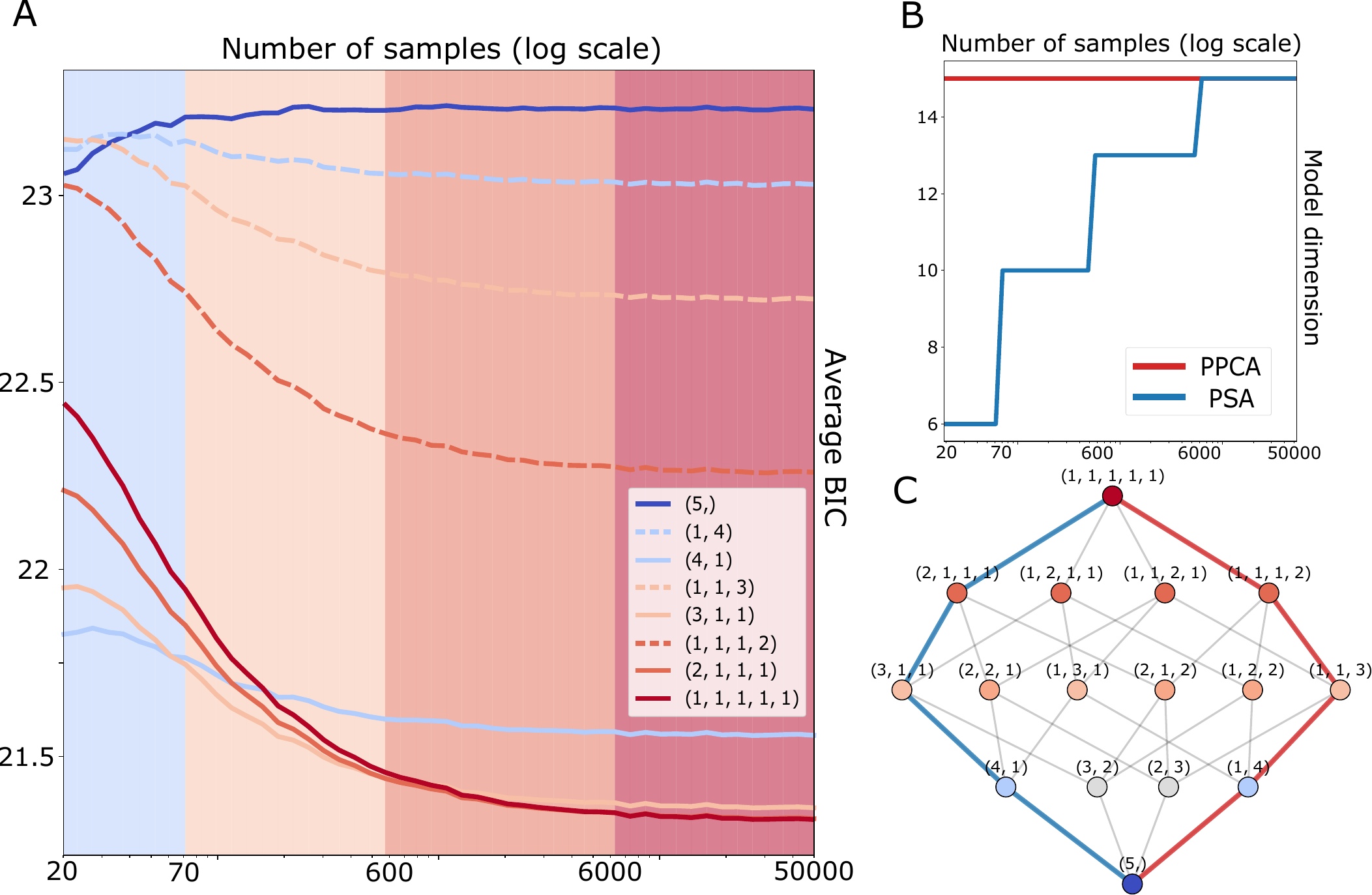}
    \caption{PSA model selection using the BIC for an increasing number of available samples.
        (A) Each curve represents the average BIC of a given PSA model over several independent experiments. The lowest curve at a given $n$ (horizontal coordinate) therefore corresponds to the most selected model.
        The curves corresponding to PPCA models are dashed.
        The curve color is related to the number of free parameters, from low (blue) to high (red). 
        The background color then corresponds to the most selected model at a given sample size.
        For instance, we can see that for $n \in [20, 70]$ (light blue), the model that is the most selected is ${\gamma} = (4,1)$. For $n \in [70, 600]$ (light orange), it is ${\gamma} = (3,1,1)$. For $n \in [600, 6000]$ (orange), it is ${\gamma} = (2,1,1,1)$. And for $n \in [6000, 50000]$ (red), it is ${\gamma} = (1,1,1,1,1)$.
        (B) Comparison of the complexities of the mostly selected models within the whole PSA family (blue) and within the PPCA family only (red).
        (C) PSA Hasse diagram. The blue curve corresponds to the trajectory followed by the optimal PSA selected model as the number of samples increases. We could expect that the PPCA models on the right follow the same kind of trajectory (in red), but it actually only stays on the top node as the other available models do not fit well the data distribution.
		}
        \label{appfig:BIC_traj}
\end{figure}
First, on the BIC plots, we can see that for $n \leq 6000$, PSA discloses a whole family of models that better explain the observed data than PPCA.
This shows that even for a very large number of samples with respect to the dimension, distinguishing the first eigenvalues and eigenvectors like PPCA does is not justified.
Second, on the complexity plots, we can see that PPCA mostly selects the full covariance model for any sample size, while PSA finds less complex models along the whole trajectory.
Moreover, interestingly, we note the consistent increase of model complexity with the number of samples. We deduce that as the sample size increases, PSA can more confidently distinguish the sample eigenvalues.
Third, on the Hasse diagram, we can see that PSA follows a trajectory as the number of available samples increases, which recalls the kind of subfamily generated by the hierarchical clustering strategy (cf. Fig~\ref{appfig:hierarchical_clustering}).
To conclude, we see on this synthetic example that PSA achieves a better complexity/goodness-of-fit tradeoff than PPCA in a wide range of sample sizes by equalizing the highest eigenvalues.

\subsection{{Statistical power of the relative eigengap}}
{
The hypothesis testing framework may be quite insightful in order to evaluate the quality of the proposed methodology.
To that extent, let us consider a dataset $(x_i)_{i=1}^n \sim \mathcal{N}(0, \diag{\lambda_1, \lambda_2})$ sampled from a two-dimensional Gaussian distribution with covariance eigenvalues $\lambda_1 \geq \lambda_2$, separated by a relative eigengap $\delta$ (i.e. $\lambda_2 = \lambda_1 (1-\delta)$). 
The null hypothesis is $\delta = 0$ (the eigenvalues are equal), and the alternative hypothesis is $\delta > 0$ (the eigenvalues are distinct).
Let $\ell_1 \geq \ell_2$ be the sample eigenvalues.
Using our relative eigengap condition~\eqref{eq:releigengap_threshold}---which itself somewhat fixes a significance level---we aim to evaluate the statistical power of the relative eigengap as a function of $\delta$ (effect size) and $n$ (sample size).}

{Let us first consider the case $\delta = 0$, with $\lambda_1 = \lambda_2 = 1$. We plot the percentage of correct identifications of isotropy in Fig~\ref{fig:R1Q1} (left).}
\begin{figure}
    \centering
    \includegraphics[width=\linewidth]{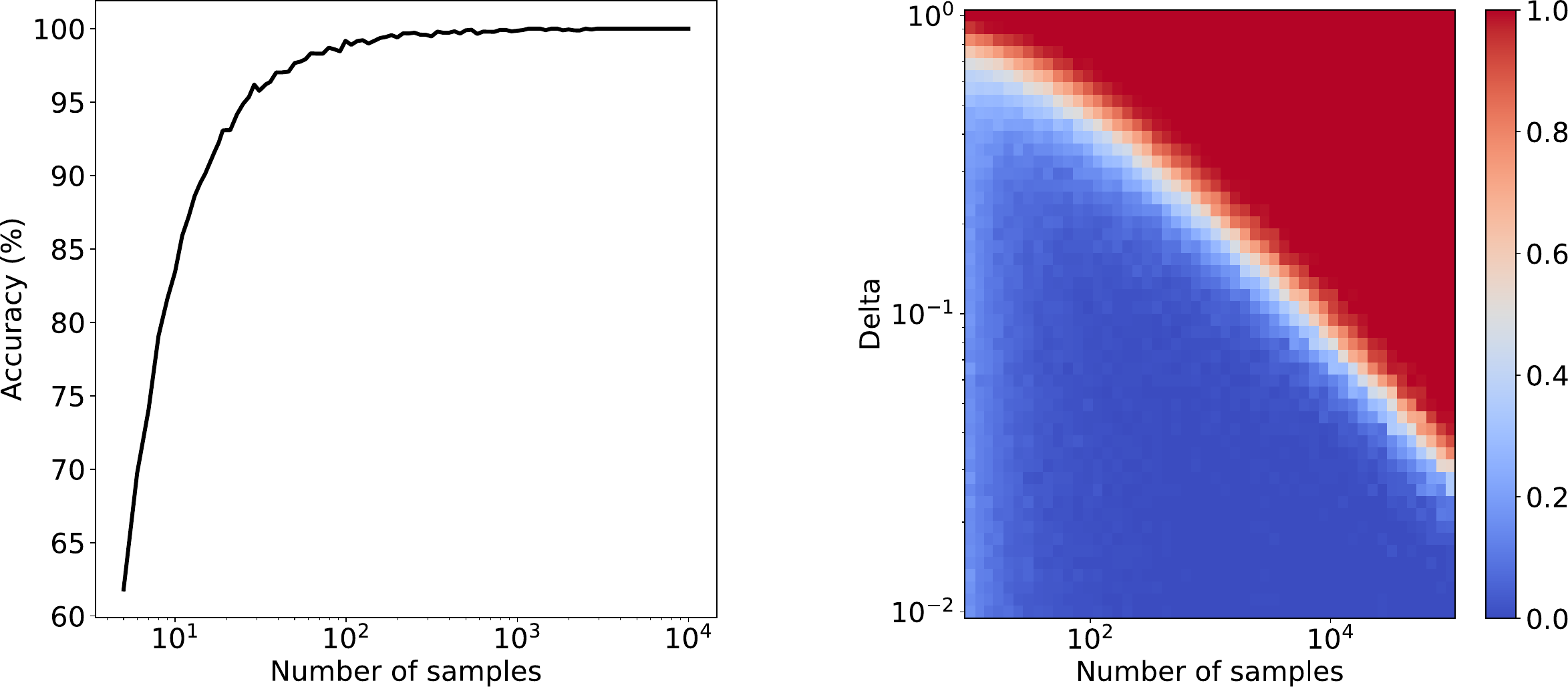}
    \caption{{Relative eigengap-based eigenvalue-equality testing under the two-dimensional Gaussian dataset $(x_i)_{i=1}^n \sim \mathcal{N}(0, \diag{1, 1-\delta})$.
    (Left) Percentage of correct identification of \textit{isotropy} ($\delta=0$) by our relative eigengap criterion for increasing $n$. (Right) Percentage of correct identification of \textit{anisotropy} ($\delta>0$) for increasing $n$ and $\delta$.}}
    \label{fig:R1Q1}
\end{figure}
{
We see that the accuracy increases with the number of samples and goes asymptotically to $100\%$; the relative eigengap condition gets more than $90\%$ accuracy for $n\geq 15$ and more than $95\%$ accuracy for $n\geq 27$.
Let us now consider the case $\delta > 0$, with $\lambda_1 = 1$ and $\lambda_2 = \lambda_1 (1 - \delta)$. For increasing $\delta$ and $n$, we plot in Fig~\ref{fig:R1Q1} (right) the percentage of correct identifications of anisotropy (statistical power) with the relative eigengap condition. We can see a sharp transition between the ``small $\delta$ small $n$'' regime where our relative eigengap condition always favors isotropic models whereas the true model has distinct eigenvalues, and the ``large $\delta$ large $n$'' regime where our relative eigengap condition always rightly favors anisotropic models. While this isotropic model misspecification in the ``small $\delta$ small $n$'' regime may sound fatal, we will see in the next subsection that it may actually have (very) positive consequences.}

\subsection{{Bias and variance of the PSA estimator}}
{
Intuitively, an expected outcome of equalizing eigenvalues (PSA) instead of inferring them individually (PPCA) is that the bias of the underlying estimator increases while the variance decreases. To assess this bias--variance tradeoff, we consider a dataset $(x_i)_{i=1}^n \sim \mathcal{N}(0, \diag{\lambda_1, \lambda_2})$ sampled from a two-dimensional Gaussian distribution with covariance eigenvalues $\lambda_1 \geq \lambda_2$, separated by a relative eigengap $\delta$ (i.e. $\lambda_2 = \lambda_1 (1-\delta)$). 
Let $\ell_1 \geq \ell_2$ be the sample eigenvalues and $v_1 \perp v_2$ some associated sample eigenvectors. 
We want to evaluate the average and standard Frobenius errors between the estimated covariance matrix and the true one: 
\begin{equation}
    \norm{\hat\Sigma - \diag{\lambda_1, \lambda_2}}_F,
\end{equation}
with $\hat\Sigma = \ell_1 v_1 {v_1}\T + \ell_2 v_2 {v_2}\T$ under the PPCA model and $\hat\Sigma = \frac{\ell_1 + \ell_2}{2} \lrp{v_1 {v_1}\T + v_2 {v_2}\T} = \frac{\ell_1 + \ell_2}{2} I_2$ under the PSA model.
}

\subsubsection{{Finite-sample simulations}}
{Let us first consider the case $\delta = 0$, with $\lambda_1 = \lambda_2 = 1$. We plot in Fig~\ref{fig:R2Q2} (top-left) the average and standard Frobenius errors for increasing $n$.}
\begin{figure}
    \centering
    \includegraphics[width=\linewidth]{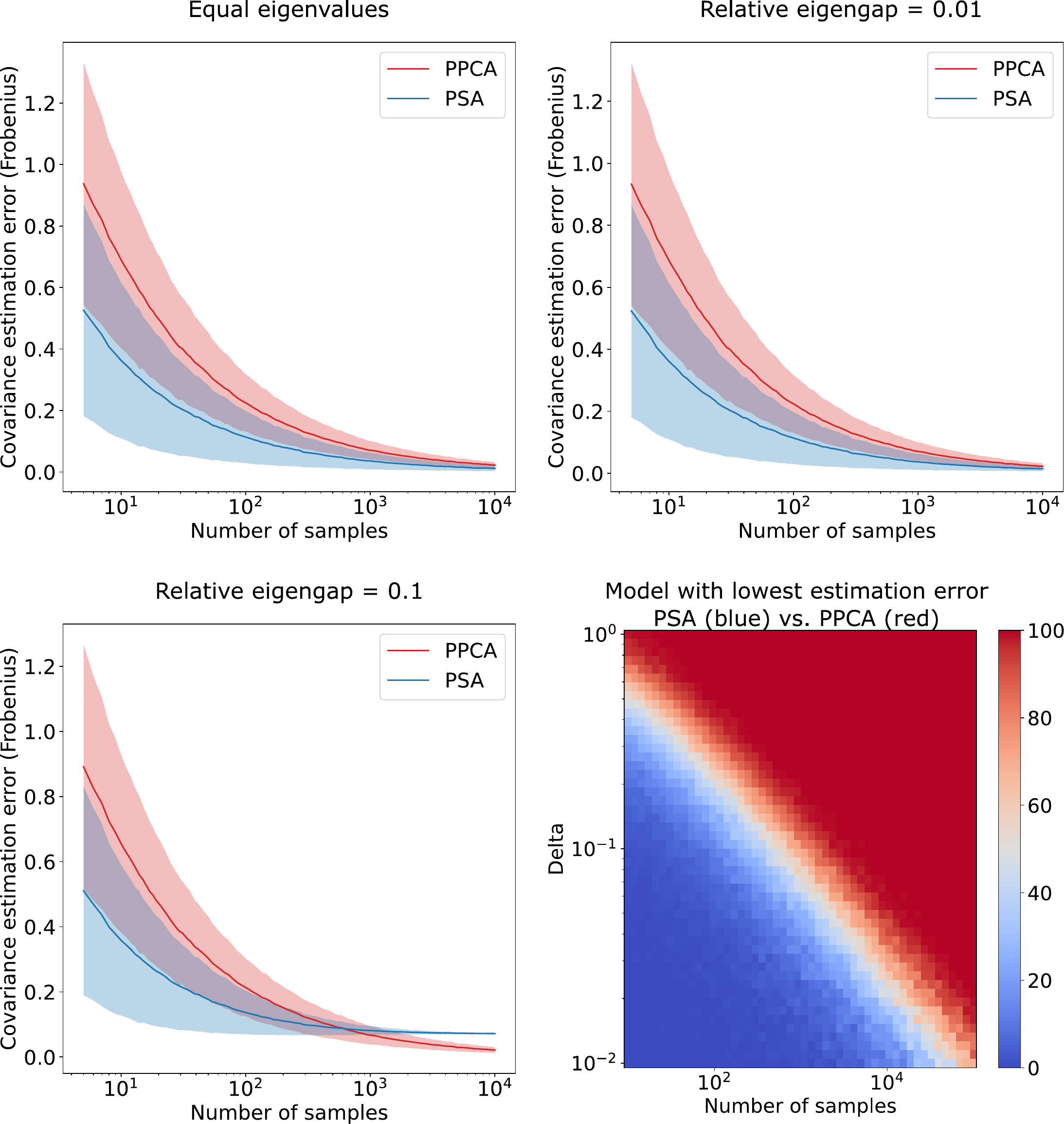}
    \caption{
    {
    Comparison of PPCA ($\lambda_1\geq\lambda_2$) and PSA ($\lambda_1=\lambda_2$) covariance estimation error $\|{\hat\Sigma - \diag{\lambda_1, \lambda_2}}\|_F$ under the two-dimensional Gaussian dataset $(x_i)_{i=1}^n \sim \mathcal{N}(0, \diag{1, 1-\delta})$. 
    (Top-left) Covariance estimation error for increasing $n$ and $\delta = 0$. 
    (Top-right) Covariance estimation error for increasing $n$ and $\delta = 0.01$. 
    (Bottom-left) Covariance estimation error for increasing $n$ and $\delta = 0.1$. 
    (Bottom-right) Lowest estimation error between PPCA (red) and PSA (blue) for increasing $n$ and $\delta$.
    }
    }
    \label{fig:R2Q2}
\end{figure}
{We see that the PSA model has a lower average estimation error for all $n$ and a lower variance too. Both estimation errors tend to $0$ asymptotically.}

{Let us  now consider the case where the two population eigenvalues are separated by a relative eigengap $\delta$, i.e. $\lambda_1 = 1$ and $\lambda_2 = \lambda_1 (1 - \delta)$. For increasing $n$, we plot in Fig~\ref{fig:R2Q2} (top-right and bottom-left) the average Frobenius errors of both methods for $\delta=0.01$ and $\delta=0.1$. 
While the variance of PSA is always lower than the variance of PPCA, both going to 0 asymptotically, we can now observe a bias in the PSA model: the PPCA error goes to $0$ asymptotically while the PSA error converges to a larger value. 
All the previous observations are quite natural---and they will be justified with simple theoretical insights later in this subsection.}

{
What is positively surprising is that when the number of samples is ``not-so-large'', the PSA estimator actually achieves a smaller error than the PPCA estimator, although being misspecified. This phenomenon is perhaps even better illustrated on the bottom-right plot of Fig~\ref{fig:R2Q2}, depicting the number of times the PPCA model yields a smaller estimation error than the PSA model for different $(n, \delta)$ values. We see that the PSA model almost surely yields a lower covariance estimation error than the PPCA model in the ``small $\delta$ small $n$'' regime.} 

{
This outcome nuances the results of the preceding subsection, which gave the impression that the PSA methodology was not suited to the ``small $\delta$ small $n$'' regime. Although the PSA models are misspecified (they assume equal eigenvalues while the true ones are distinct), the parsimony induced by equalizing the close eigenvalues actually yields smaller estimation errors. Interestingly, PPCA needs quite a lot of samples to outperform PSA's covariance estimation, although the latter is misspecified compared to the former.
}  %

{
Hence, this experiment shows that the true covariance matrix does not need to have repeated eigenvalues to make our PSA models interesting. They reduce \textit{both} the bias and the variance for small-to-moderate sample sizes.
}

\subsubsection{{Asymptotic theoretical insights}}
{The literature on asymptotic distributions of principal components (see~\cite[Section~3.6]{jolliffe_principal_2002} for a quick overview) enables us to get simple theoretical insights on the previous observations.}

{For instance, if we assume that the population eigenvalues are distinct ($\lambda_1 > \lambda_2$), then the asymptotic distribution of the (ordered decreasing) sample eigenvalues ($\ell_1 \geq \ell_2$) is given in Eq~(3.10) of Anderson's seminal paper~\cite{anderson_asymptotic_1963}:
\begin{equation}
    \begin{cases}
        \sqrt{n} (\ell_1 - \lambda_1) \sim \mathcal{N}(0, 2\lambda_1^2),\\
        \sqrt{n} (\ell_2- \lambda_2) \sim \mathcal{N}(0, 2\lambda_2^2).
    \end{cases}
\end{equation}
Therefore, one gets
\begin{equation}
    \begin{cases}
        \sqrt{n}\lrp{\frac{\ell_1 + \ell_2}{2} - \lambda_1} \sim \mathcal{N}\lrp{-\frac{\lambda_1 - \lambda_2}{2}, \frac{\lambda_1^2 + \lambda_2^2}{2}},\\
        \sqrt{n}\lrp{\frac{\ell_1 + \ell_2}{2} - \lambda_2} \sim \mathcal{N}\lrp{+\frac{\lambda_1 - \lambda_2}{2}, \frac{\lambda_1^2 + \lambda_2^2}{2}}.
    \end{cases}
\end{equation}
Consequently, as intuited with the experiments, the PSA estimator is biased while the PPCA estimator is not. Moreover, the PSA estimator has a lower variance than the PPCA estimator. The same reasoning generalizes seamlessly to any dimension $p$ and grouping of eigenvalues $\gamma\in\mathcal{C}(p)$. The PSA model, which block-averages the sample eigenvalues, is biased, but its variance is reduced quadratically with respect to the sizes of the blocks. More precisely, the variance is divided by $\gamma_k^2$ respectively for each block.
}

{
Let us now assume that the population eigenvalues are equal ($\lambda_1 = \lambda_2 \coloneqq \lambda$). Then the asymptotic distribution of the (ordered-decreasing) sample eigenvalues is derived in~\cite[Eq~(2.12)]{anderson_asymptotic_1963}. Denoting $h \coloneqq \sqrt{n}(\ell - \lambda)$, one has
\begin{equation}
    p(h_1, h_2) = \frac{1}{\sqrt{32 \pi} \lambda^3} e^{-\frac{h_1^2 + h_2^2}{4\lambda^2}} (h_1-h_2) \, \mathbf{1}_{\{(h_1, h_2)\in\R^2\colon h_1 > h_2\}} (h_1, h_2).
\end{equation}
Using changes of variables and truncated Gaussian integrals, one gets the following moments:
\begin{equation}
    \begin{cases}
        \int_{h_1>h_2}   \frac{h}{\sqrt{32 \pi} \lambda^3} e^{-\frac{h_1^2 + h_2^2}{4\lambda^2}} (h_1-h_2) \, \mathrm{d}h_1 \, \mathrm{d}h_2 = \lrp{+\sqrt{\frac{\pi}{2}}\lambda, -\sqrt{\frac{\pi}{2}}\lambda},\\
        \int_{h_1>h_2}   \frac{h^2}{\sqrt{32 \pi} \lambda^3} e^{-\frac{h_1^2 + h_2^2}{4\lambda^2}} (h_1-h_2) \, \mathrm{d}h_1 \, \mathrm{d}h_2 = (3 {\lambda}^{2}, 3 {\lambda}^{2}).
    \end{cases}
\end{equation}
Therefore, one has
\begin{equation}
    \begin{cases}
        \mathbb{E}[\sqrt{n}(\ell_1 - \lambda)] &= + \sqrt{\frac{\pi}{2}} \lambda,\\
        \mathbb{E}[\sqrt{n}(\ell_2 - \lambda)] &= - \sqrt{\frac{\pi}{2}} \lambda,\\
        \mathbb{V}[\sqrt{n}(\ell_1 - \lambda)] &= (3 - \frac{\pi}{2})\lambda^2,\\
        \mathbb{V}[\sqrt{n}(\ell_2 - \lambda)] &= (3 - \frac{\pi}{2})\lambda^2.
    \end{cases}
\end{equation}
We see that the (ordered decreasing) sample eigenvalues are biased.}

{
Conversely, using Eq~(3.10) of Anderson's seminal paper~\cite{anderson_asymptotic_1963}, one gets 
\begin{equation}
    \sqrt{n}\lrp{\frac{\ell_1 + \ell_2}{2} - \lambda} \sim \mathcal{N}\lrp{0, {\lambda^2}}.
\end{equation}
Hence the PSA estimator is not only unbiased but also has lower variance than the PPCA estimator.
}

{The last result on the PSA estimator generalizes seamlessly to any dimension $p$ and grouping of eigenvalues $\gamma\in\mathcal{C}(p)$, where the variance is divided by $\gamma_k$ (cf.~\cite[Eq~(3.10)]{anderson_asymptotic_1963}). The penultimate result on the PPCA estimator may generalize to higher dimensions but the formulas would be much more complicated.}

\subsection{{Model selection accuracy of hierarchical clustering algorithm}}
{
Let us now evaluate the quality of Algorithm~\ref{alg:hierarchical}, in terms of model selection accuracy. More precisely, given a synthetic PSA-distributed dataset, let us estimate the probability that Algorithm~\ref{alg:hierarchical} recovers the correct eigenvalue multiplicities.
Let $(x_i)_{i=1}^n \sim \mathcal{N}(0, \diag{\lambda_1 I_{20}, \lambda_2 I_{20}, \lambda_3 I_{10}})$ be a dataset with $n$ points sampled from a multivariate Gaussian distribution with $p=50$ and covariance eigenvalues $\lambda_1=10$ (of multiplicity $20$), $\lambda_2 = 10 \times (1-\delta)$ (of multiplicity $20$) and $\lambda_3 = 10 \times (1-\delta)^2$ (of multiplicity $10$). The idea of such a covariance profile is to have three blocks of eigenvalues, with constant inter-block relative eigengap $\delta$.
}

{We report in Fig~\ref{fig:R2Q3} (top-left, top-right, bottom-left) some typical sample eigenvalue profiles generated from this model. We can see that for $\delta =0.5$ and $n=100$, the three groups of eigenvalues are not visually identifiable. As $n$ increases, the three groups get more and more separated. Let us note that the top sample eigenvalue sometimes has a relatively large difference with the first block of eigenvalues, which could lead model selection methods to separate them.}

{We now report in  Fig~\ref{fig:R2Q3} (bottom-right) the percentage of accurate model selection with our hierarchical eigenvalue clustering method (Algorithm~\ref{alg:hierarchical}), as a function of $n$ and $\delta$.}
\begin{figure}
    \centering
    \includegraphics[width=\linewidth]{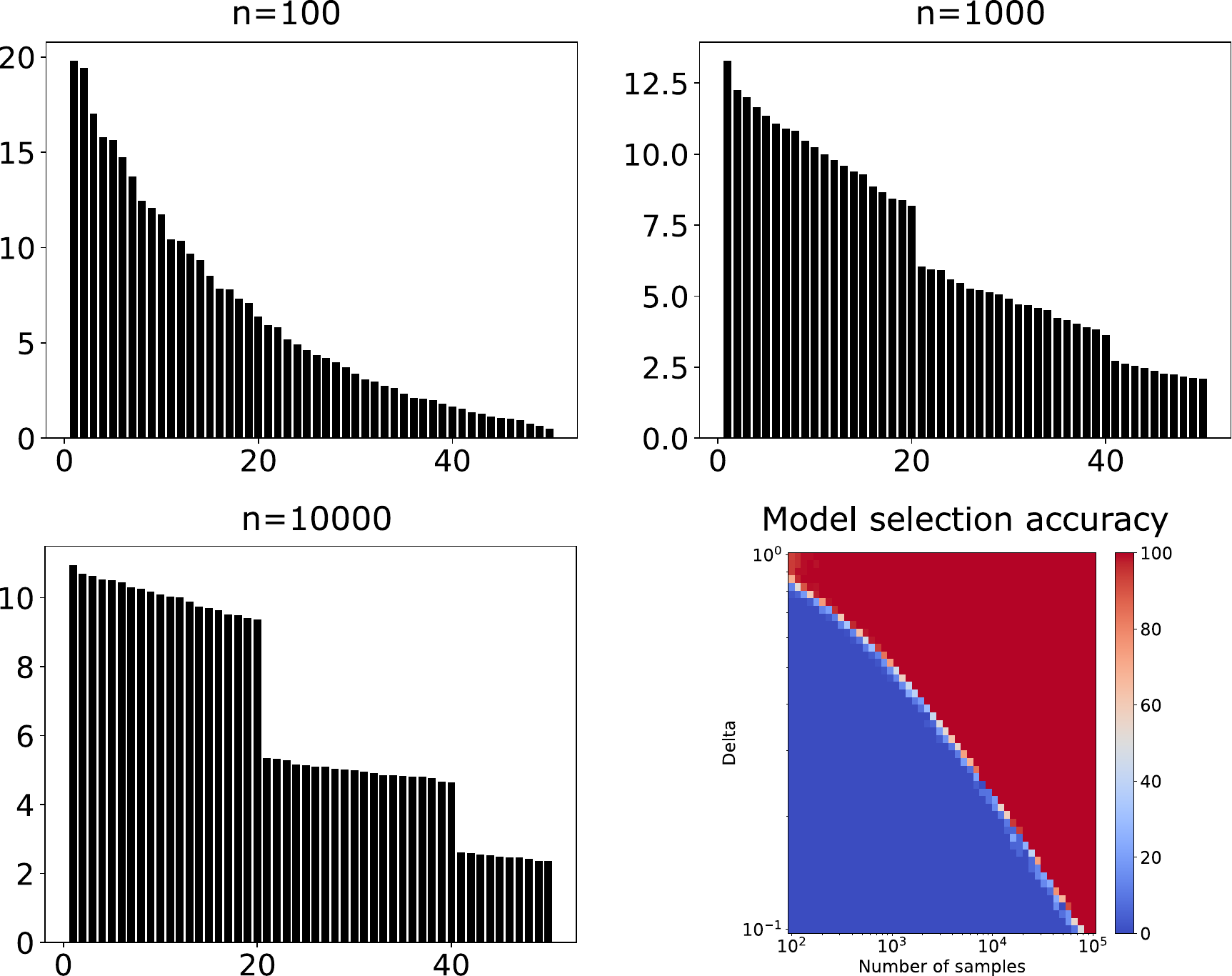}
    \caption{
    {
    Algorithm~\ref{alg:hierarchical}'s ability to recover the true eigenvalue multiplicities under the three-block population covariance matrix $\diag{\lambda_1 I_{20}, \lambda_2 I_{20}, \lambda_3 I_{10}}$. 
    (Top-left) Sample eigenvalue profile for $\delta = 0.5$ and $n = 100$.
    (Top-right) Sample eigenvalue profile for $\delta = 0.5$ and $n = 1000$.
    (Bottom-left) Sample eigenvalue profile for $\delta = 0.5$ and $n = 10000$.
    (Bottom-right) Percentage of accurate model selection with Algorithm~\ref{alg:hierarchical}, for increasing $n$ and $\delta$.
    }
    }
    \label{fig:R2Q3}
\end{figure}
{We can once again see a sharp transition in terms of model selection accuracy, from 0\% for the ``small $\delta$ small $n$'' regime to 100\% for the ``large $\delta$ large $n$'' regime.}

\section{{Alternative methods for grouping eigenvalues and perspectives}}\label{appsec:alternatives}
{
While our proposed BIC-based methodology for grouping the eigenvalues is certainly practical, it may seem rather heuristic than relying on strong theoretical foundations. This section discusses some alternative methods and perspectives to identify the curse of isotropy. 
We thank the anonymous reviewers for suggesting us to address these perspectives.}

{
\begin{remark}
We initially opted for a BIC-based methodology due to the ubiquity of such criteria in data science. An interesting anecdote is that the default method for estimating PCA's intrinsic dimension in one of the most used data science libraries (scikit-learn~\cite{pedregosa_scikit-learn_2011}) is Minka's penalized-likelihood~\cite{minka_automatic_2000}, which can be seen as a refinement of the BIC. 
Therefore, we believe that the BIC and related model selection criteria are quite widespread among practitioners, hence the practical interest of our methodology. Moreover, such criteria do enjoy theoretical foundations and guarantees~\cite{kass_bayes_1995,bai_consistency_2018}.
\end{remark}
}

\subsection{{Continuous relaxation of model selection}}
{
A natural alternative to (discrete) model selection for PSA is the penalized-likelihood approach, with a continuous penalty enforcing \textit{sparsity} of the \textit{eigengaps}, i.e. equal eigenvalues. We investigated this idea in a follow-up conference paper~\cite{szwagier_eigengap_2025}. The main findings are summarized in the following paragraphs.}

{
First, we derive an $\ell^1$-relaxation of the PSA model selection methodology. More precisely, the Bayesian information criterion~\eqref{appeq:BIC} used for model selection is rewritten as a penalized log-likelihood, where the penalty is a thoroughly-derived $\ell^0$-norm of pairwise distances between eigenvalues: the \textit{eigengaps}.
The BIC is then relaxed with $\ell^1$-norms, which results in a continuous optimization criterion. Such an approach has several advantages compared to a heuristic penalization. For instance, the regularization tuning hyperparameter $\alpha\in\R_+$ (which is often present in penalized optimization problems) is unique and automatically determined by the BIC ($\alpha = \ln n$). Moreover, the relaxed problem enjoys the statistical guarantees of the BIC whenever the relaxation is tight.
}

{Second, although penalizing the differences between \textit{adjacent}-eigenvalues-{only} seems intuitive, we show that the accurate way to relax the parsimony constraints is by penalizing the differences between \textit{all} eigenvalues---adjacent and non-adjacent. The justification is a bit technical, but in summary, we show that the number of covariance parameters related to the repeated eigenvalues, $d$, can be written as an $\ell^0$-norm of differences between adjacent eigenvalues, while the number of covariance parameters related to the flag of eigenspaces, $p(p-1) / 2 - \sum_{k=1}^d \gamma_k (\gamma_k - 1)/2$, can be written as an $\ell^0$-norm of differences between all pairs of eigenvalues. 
Hence, although penalizing the differences between adjacent eigenvalues seems intuitive, we show that accounting for the covariance eigenspaces requires to add the non-adjacent eigenvalues too, which importantly increases the ``strength'' of the penalty. 
This subtlety is actually very important, since it really enables to create large clusters of eigenvalues---therefore decreasing quadratically the number of parameters---while penalties on the adjacent eigengaps only may just equalize isolated pairs of eigenvalues.
}

{Third, we believe that the absolute distance between eigenvalues is not the right metric to use for the penalty. Indeed, we conjecture that the critical points of the penalized-likelihood objective function (when the penalty is on the absolute differences between eigenvalues) necessarily correspond to \textit{isotropic} covariance matrices. Hence, we decide to use relative eigengaps instead of absolute eigengaps in the methodology.
}

{Fourth, the final projected-gradient-descent algorithm that we propose unexpectedly draws interesting links with some classical covariance shrinkage methods~\cite{ledoit_quadratic_2022}. It notably suggests that parsimony in covariance matrices tends to ``mutually attract'' the eigenvalues, which is a well-known side effect of covariance shrinkage methods.
Moreover, our eigengap sparsity draws interesting links with the elasso method from Tyler and Yi~\cite{tyler_lassoing_2020} and follow-up works~\cite{basiri_fusing_2019}.
}

\subsection{{Bootstrap-based stability analysis}}

{
In view of the \textit{intersample variability}-related motivations for principal subspace analysis (cf. Section~\ref{sec:intro}), some alternative methodologies to detect the curse of isotropy based on bootstrapping may appear as natural. This subsection details two bootstrap-based methodologies to assess the stability of the principal components across independent samples. The first idea is based on eigenvalue confidence intervals: if two adjacent eigenvalues' confidence intervals intersect, then we equalize them. The second idea is based on eigenvectors variability: if one eigenvector ``fluctuates'' significantly, then we should merge it with the adjacent eigenvectors. 
}

{
First, the idea of confidence interval intersection for the eigenvalues is actually closely related to North's rule-of-thumbs~\cite{north_sampling_1982}, that we evoke in Section~\ref{sec:related} and discuss here in subsection~\ref{appsubsec:north}. 
Indeed, under the Gaussian assumption, one can derive the asymptotic normal law of the sample eigenvalues ($\ell_j\sim \mathcal{N}(\lambda_j, 2\lambda_j^2/n)$) and \textit{exactly} rewrite the intersection of the $95\%$ confidence intervals as a relative eigengap inequality (cf. Proposition~\ref{appprop:releigengap_North}). 
The curve of the threshold is plotted in Fig~\ref{fig:releigengap_curves} (NRT-2, for the $2\sigma$ confidence intervals). We see that the threshold is larger than our BIC-based threshold for small $n$ and smaller for large $n$, with a transition appearing at $n \approx 100$. This implies that the confidence-interval based approach equalizes more eigenvalues in the small-to-moderate sample regime and less eigenvalues in the large sample regime. Since, as shown in~\cite{north_sampling_1982}, the fluctuations of eigenvectors are first-order-proportional to the inverse of the eigengaps, then we believe that similar conclusions can be made for the idea on the fluctuation of the eigenvectors.
}

{
Now, since we are interested in the non-asymptotic regime, let us actually conduct the bootstrap experiments in a very simple case, that is $X\coloneqq (x_i)_{i=1}^n \sim \mathcal{N}(0, \diag{\lambda_1, \lambda_2})$ with $\lambda_1 = 1$ and $\lambda_2 = \lambda_1 (1 - \delta)$.
The intersection of the $95\%$ confidence intervals is relatively straightforward to implement. In contrast, the formulation of the eigenvector fluctuation idea is much more open, therefore we detail it hereafter. 
}

{
Let $(X'_l)_{l=1}^{n_\mathrm{res}}\in\R^{n_\mathrm{res}\times n\times p}$ correspond to $n_\mathrm{res}$ $n$-samples with replacement from $X\in\R^{n\times p}$. Let $v_l \in \R^p$ denote the leading eigenvector of the covariance matrix associated with the dataset $X'_l$.
Motivated by the classical principal-angle-based subspace distances (cf.~\cite[Section~2]{ye_schubert_2016} for instance), we define the following quantity as the \textit{fluctuation statistic}:
\begin{equation}
    \sigma\coloneqq \sqrt{\frac{1}{n^2}\sum_{l, l'=1}^{n_\mathrm{res}}\arccos({v_l}\T v_{l'})^2}.
\end{equation}
We decide to equalize the two eigenvalues when $2\sigma > \pi / 4$. The intuition is that the eigenvectors' orientation lies between $0$ and $\pi/2$, so that there will likely be a strong overlap between the two eigenvectors when $2\sigma > \pi / 4$.
The numerical tests for these two ideas are reported in Fig~\ref{fig:R2Q4}.
}
\begin{figure}
    \centering
    \includegraphics[width=\linewidth]{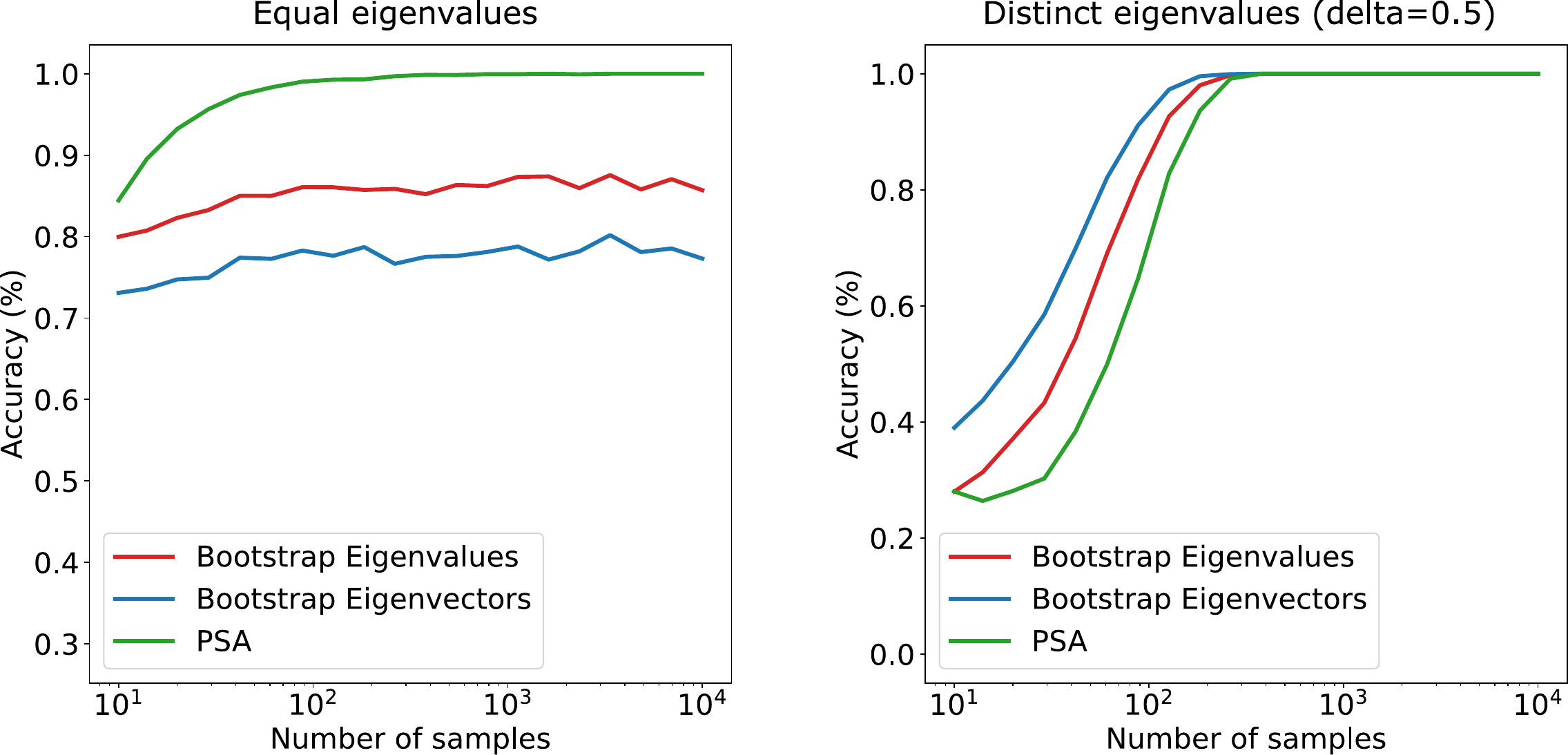}
    \caption{
    {Comparison of three eigenvalue grouping heuristics: relative eigengap, eigenvalue bootstrap and eigenvector bootstrap under the two-dimensional Gaussian dataset $(x_i)_{i=1}^n \sim \mathcal{N}(0, \diag{1, 1-\delta})$.
    (Left) Percentage of correct identification of \textit{isotropy} ($\delta=0$) by the three heuristics for increasing $n$. (Right) Percentage of correct identification of \textit{anisotropy} ($\delta=0.5$) by the three heuristics for increasing $n$.}
    }
    \label{fig:R2Q4}
\end{figure}

{
We can see that the PSA model with the BIC tends to more often favor parsimonious models than the bootstrap-based methods. This somewhat matches the asymptotic theory (cf. Fig~\ref{fig:releigengap_curves}) that the BIC favors more parsimonious models than North's rule of thumbs with $95\%$ confidence intervals. The bootstrap-based methods are naturally much longer to run (proportionally to the number of resamples $n_\mathrm{res}$), but they are quite practical and distribution-agnostic. Another issue with the bootstrap-based methods is that they rely on the choice of the width of the confidence intervals, which will obviously influence the parsimony of the selected model, while the BIC-based method is hyperparameter-free.
}

\subsection{{Bayesian extensions}}
{ 
We hereafter list some ideas of prior distributions for the type $\gamma\coloneqq(\gamma_1, \dots, \gamma_d)\in\mathcal{C}(p)$, the (ordered-decreasing) eigenvalues $\lambda_1, \dots, \lambda_d$ and the (mutually-orthogonal) frames $Q_1, \dots, Q_d$.
}

{The most natural prior for the type $\gamma\in\mathcal{C}(p)$ is the uniform prior over the (discrete) family of PSA models, i.e. $~{p(\gamma) \propto 1}$ (the normalizing constant being $\#\mathcal{C}(p)^{-1} = {2^{1-p}}$). An alternative prior is the uniform prior over a subset of PSA models. For instance, one can bound the complexity of the candidate models by considering priors of the form $p(\gamma) \propto \mathbf{1}_{d \leq d^*}(\gamma)$ for a given $d^*\in[1, p]$, the normalizing constant being $\bigl({\sum_{k=0}^{d^*-1} \binom{p-1}{k}}\bigr)^{-1}$. Such a prior imposes an upper bound on the number of eigenvalue blocks $d$, which is equivalent to considering a few lower floors in the Hasse diagram from Fig~\ref{fig:hasse_complexity}.
One can also bound the complexity of the model with priors of the form $p(\gamma) \propto \mathbf{1}_{p - \gamma_d \leq q^*}(\gamma)$ for a given $q^*\in[0, p-1]$, the normalizing constant being ${2^{-q^*}}$. Such a prior imposes an upper bound on the intrinsic dimension $q$. Finally, we can consider non-uniform priors putting more weights towards simpler models, like $p(\gamma) \propto \exp(-d)$ or $p(\gamma) \propto \exp(-(p-\gamma_d))$.
Let us point that in each case, we have the normalization constant in closed-form since we can easily---up to basic combinatorics---enumerate the candidate models.}

{There are plenty of possible priors for the eigenvalues. In the celebrated paper of Minka~\cite{minka_automatic_2000}, the prior is a scaled-inverse chi-squared distribution: $p(\lambda) \propto |\diag{\lambda}|^{-(\alpha+2)/2 }\exp(-(\alpha/2)\operatorname{tr}({\diag{\lambda}^{-1}}))$, where $\alpha$ is a hyperparameter controlling the ``sharpness'' of the prior. This choice is motivated by the use of a conjugate prior for the Gaussian likelihood of the covariance matrix, to facilitate the computations.
This automatically yields decreasingly-ordered eigenvalues for the maximum a posteriori estimate.
}

{The most natural prior for the frames is a uniform prior on the flag manifold, i.e. $(Q_1, \dots, Q_d) \sim \mathcal{U}(\Fl(\gamma))$. Since we are on Riemannian manifolds, the notion of ``uniformity'' is induced by the Riemannian measure, which itself is defined via the Riemannian metric. If we take the canonical metric, similarly as in the celebrated paper of Minka~\cite{minka_automatic_2000}---which is itself based on~\cite{james_normal_1954} and which involves Stiefel manifolds---then we can compute explicitly the normalizing constant, which is the reciprocal area of the flag manifold~\cite{boya_volumes_2003}. The latter is a generalization of the volume of Stiefel and Grassmann manifolds, via the quotient structure~\eqref{eq:quotient}.
We can also consider non-uniform priors on the frames, like matrix Von Mises--Fisher and Bingham distributions~\cite{khatri_von_1977,jupp_maximum_1979,hoff_hierarchical_2009,pal_conjugate_2020}, to shrink the flag of eigenspaces towards central values.}

{One can finally consider full covariance models with priors favoring equal eigenvalues. A natural prior for that is the reference prior of Yang and Berger~\cite{yang_estimation_1994} $p(\lambda) = c[|\diag{\lambda}|\prod_{i<j}(\lambda_i - \lambda_j)]^{-1}$. This prior puts more mass in the regions of eigenvalue equality~\cite{pourahmadi_covariance_2011}. Another natural idea is Wigner's surmise, which is directly on the ``spacing'' $\delta$ between eigenvalues $p(\delta)={\frac {\pi \delta}{2}}e^{-\pi \delta^{2}/4}$. In a similar vein, one could also consider Laplace or exponential distributions (similarly as in the seminal LASSO paper~\cite[Section~5]{tibshirani_regression_1996}) on the eigengaps, in order to favor exact equality of eigenvalues.}

\section{Information about datasets}\label{appsec:data}
In this section, we give a few more details about the data used for the experiments.

\subsection{Natural image patches}
In this experiment, we consider 10 flower images from the ImageNet database~\cite{deng_imagenet_2009}. Those were downloaded from Kaggle (\url{https://www.kaggle.com/datasets/prasunroy/natural-images}) and extracted from \texttt{natural\_images/flower/} folder, from \texttt{flower\_0000.jpg} up to \texttt{flower\_0009.jpg}.

\subsection{Eigenfaces}
In this experiment, we consider 31 digital images from the CMU Face Images database~\cite{mitchell_cmu_1997}. Those were downloaded from Kaggle (\url{https://www.kaggle.com/datasets/raviprakash22/cmu-face-images}) and extracted from the folder \texttt{faces/faces/choon}. We only extracted the $(60, 64)$ images, corresponding to all the files ending with \texttt{\_2.pgm}.

\subsection{Structured data}
For the structured data experiment (cf. Fig~\ref{fig:exp_rotation}) and the relative eigengap tables (cf. Fig~\ref{fig:releigengap_UCI} and Fig~\ref{appfig:releigengaps_real}), we consider data from the UCI Machine Learning Repository (\url{https://archive.ics.uci.edu/}): Ionosphere~\cite{sigillito_ionosphere_1989}, Wine~\cite{aeberhard_wine_1991}, Wisconsin~\cite{wolberg_breast_1995}, Glass~\cite{german_glass_1987}, Iris~\cite{fisher_iris_1936}, Spambase~\cite{hopkins_spambase_1999}, Digits~\cite{alpaydin_optical_1998}, Covertype~\cite{blackard_covertype_1998}.
\end{document}